\documentclass[twoside]{article}
\usepackage{amsfonts, amsthm}
\usepackage[all]{xy}
\usepackage[mathcal]{eucal}
\usepackage[dvips]{graphicx}
\usepackage{pictexwd,dcpic}

\newtheorem{theorem}{Theorem}[section]
\newtheorem{lemma}[theorem]{Lemma}
\newtheorem{prop}[theorem]{Proposition}
\newtheorem{cor}[theorem]{Corollary}
\newtheorem{remark}[theorem]{Remark}

\def\l{\lambda}
\def\a{\alpha}
\def\b{\beta}
\def\d{\delta}
\def\g{\gamma}
\def\e{\epsilon}
\def\t{\tau}
\def\z{\zeta}

\def\D{N'}
\def\HD{H^{(D)}}
\def\Hg{\mathcal{H}_\g}

\def\R{\mathbb{R}}

\def\C{\mathbb{C}}

\def\Z{\mathbb{Z}}
\def\P{\mathbb{P}}

\def\T{\mathbb{T}}

\def\M{\mathcal{M}}
\def\Mba{\M_{\beta, \alpha}}
\def\Mdc{\M_{\delta, \gamma}}

\def\Thdc{\Theta_{\delta, \gamma}}
\def\Hba{H_{\beta, \alpha}}
\def\Hdc{H_{\delta, \gamma}}
\def\Oba{\Omega_{\beta, \alpha}}
\def\Rba{R_{\beta, \alpha}}

\def\P{\mathcal{P}}

\def\n{\nabla}

\def\ndc{\nabla_{d,c}}

\def\mbf2{\mathbf{2}}

\def\1N1{1 \leq k \leq N-1}

\def\zdir{\mathcal{Z}^{(D)}}
\def\thmz{\Theta_{\mathcal{Z}}}
\def\sz{S_{\mathcal{Z}}}

\def\H0bd{H_{0,\b}^{D}}

\def\fn4{\frac{N}{4}}

\begin{document}

\title{Symmetries of the periodic Toda lattice, with an application to normal forms and perturbations of the lattice with Dirichlet boundary conditions}
\author{Andreas Henrici\footnote{Supported in part by the Swiss National Science Foundation and a private fellowship}}

\maketitle

\begin{abstract}
Symmetries of the periodic Toda lattice are expresssed in action-angle coordinates and characterized in terms of the periodic and Dirichlet spectrum of the associated Jacobi matrices. Using these symmetries, the phase space of the lattice with Dirichlet boundary conditions is embedded into the phase space of a higher-dimensional periodic lattice. As an application, we obtain a Birkhoff normal form and a KAM theorem for the lattice with Dirichlet boundary conditions\footnote{2010 Mathematics Subject Classification: 37J15, 37K05, 37J40, 70H08, 70H33}.
\end{abstract}


\section{Introduction} \label{introduction}

Consider the Toda lattice with period $N$ ($N \geq 2$), 
\begin{equation} \label{pertodaham}
\dot{q}_n = \partial_{p_n} H_{Toda}, \quad \dot{p}_n = -\partial_{q_n} H_{Toda}
\end{equation}
for $n \in \Z$, where the (real) coordinates $(q_n, p_n)_{n \in \Z}$ satisfy
\begin{equation} \label{periodrandbed}
 (q_{n+N}, p_{n+N}) = (q_n, p_n) \qquad \forall \; n \in \Z
\end{equation}
and the Hamiltonian $H_{Toda}$ is given by
\begin{equation} \label{hamtodapq}
H_{Toda} = \frac{1}{2} \sum_{n=1}^N p_n^2 + \alpha^2 \sum_{n=1}^N e^{q_n - q_{n+1}}
\end{equation}
where $\a$ is a positive parameter, $\a > 0$. For the standard Toda lattice, $\a=1$. The Toda lattice was introduced by Toda \cite{toda} and studied extensively in the sequel. It is an important model for an integrable system of $N$ particles in one space dimension with nearest neighbor interaction and belongs to the family of the FPU chains, introduced and numerically investigated by Fermi, Pasta, and Ulam in their seminal paper \cite{fpu}. The integrability of the Toda lattice explains the quasi-periodic features of its solutions observed numerically in \cite{fpu} at least in this case.

We are interested in symmetries of the periodic lattice. Let $T^*\R^N$ be endowed with the canonical symplectic structure and consider the linear maps $T, S: T^*\R^N \to T^*\R^N$ given by
\begin{eqnarray}
 T: (q_1, \ldots, q_N, p_1, \ldots, p_N) & \mapsto & (q_2, q_3, \ldots, q_N, q_1, p_2, p_3, \ldots, p_N, p_1), \label{rqpdef} \\
 S: (q_1, \ldots, q_N, p_1, \ldots, p_N) & \mapsto & -(q_{N-1}, \ldots, q_1, q_N, p_{N-1}, \ldots, p_1, p_N); \label{sqpdef}
\end{eqnarray}
note that $T$ is the standard shift operator. As already discussed by Rink \cite{rink03} for arbitrary FPU chains, the maps $T$ and $S$ satisfy the relations $T^N = S^2 =$ Id and $T S = S T^{-1}$. Moreover, $T$ and $S$ are symplectic maps leaving the Hamiltonian $H_{Toda}$ invariant. The group $G_H = \langle T, S \rangle$ (a representation of the $N$-th dihedral group $D_N$) is the symmetry group of $H_{Toda}$.

Denote by Fix$(S)$ the fixed point set of $S$. Then Fix$(S)$ is the subset of all elements $(q,p)$ in $T^*\R^N$ satisfying
\begin{eqnarray}
  (q_n, p_n) = -(q_{N-n}, p_{N-n}) \;\; \forall \, 1 \leq n \leq N-1 \; \textrm{and} \; q_N = p_N = 0. \label{fixsdef}
\end{eqnarray}
In particular, on Fix$(S)$, $q_N = q_{N/2} = 0$ and $p_N = p_{N/2} = 0$. Note that on Fix$(S)$, both the center of mass coordinate $Q = \frac{1}{N} \sum_{i=1}^N q_i$ and its momentum $P = \frac{1}{N} \sum_{i=1}^N p_i$ are identically zero. Hence Fix$(S) \subseteq \{ (q,p) \in T^*\R^N | Q=0; P=0 \}$.

The main motivation for our interest in these symmetries of the periodic lattice is the fact that the map $S$ playes a crucial role in the investigation of  the lattice with \emph{Dirichlet boundary conditions} instead of (\ref{periodrandbed}), i.e. the lattice with $N'$ particles ($N' \geq 2$) and Hamiltonian
\begin{equation} \label{hamdirtodapq}
H_{Toda}^{(D)} = \frac{1}{2} \sum_{n=1}^{N'} p_n^2 + \g^2 \sum_{n=0}^{N'} e^{q_n - q_{n+1}},
\end{equation}
where $\g$ again is a positive parameter, $\g > 0$, independent from $\a$, and boundary conditions
\begin{equation} \label{dirbdycond}
 q_0 = q_{N'+1} = p_0 = p_{N'+1} \equiv 0.
\end{equation}
instead of (\ref{periodrandbed}). 
The phase space of such a lattice can be embedded into the phase space of the periodic lattice with a greater number of particles (namely $N=2N'+2$), and its image under this embedding is precisely the fixed point set Fix$(S)$, a submanifold of the entire phase space of the periodic lattice which is invariant under the evolution (\ref{pertodaham}). This embedding allows us to obtain results on the Dirichlet lattice by exploiting the properties of $S$. 

Before stating our results on the symmetries $T$ and $S$ introduced above, let us state the results on the Dirichlet lattice obtained by this embedding, since this application is the main reason for investigating the symmetries of the periodic lattice. We first construct a Birkhoff normal form of the Dirichlet Toda lattice.

\begin{theorem} \label{dirtodabnf4}
 For any fixed $\g \in \R$ and $\D \geq 2$, the Dirichlet Toda lattice admits a Birkhoff normal form. More precisely, there are (globally defined) canonical coordinates $(x_k,y_k)_{1 \leq k \leq \D}$ so that 
$\HD_{Toda}$, when expressed in these coordinates, is a function of the action variables $I_k = (x_k^2 + y_k^2)/2$ ($1 \leq k \leq \D$), $\HD_{Toda} = \Hg(I)$. Moreover, near $I=0$, $\Hg(I)$ has an expansion of the form
\begin{equation} \label{dirtodabnf4formula}
 \frac{N \g^2}{2} + \sqrt{2} \, \g \sum_{k=1}^{\D} s_k I_k + \frac{1}{16(\D+1)}\sum_{k=1}^{\D} I_k^2 + O(I^3),
\end{equation}
with $s_k = \sin \frac{k \pi}{2\D+2}$.
\end{theorem}

\begin{cor} \label{dirtodahessian}
 Let $\g > 0$ be arbitrary. Then the Hessian of $\Hg(I)$ at $I=0$ is given by
\begin{displaymath}
 d_I^2 \Hg(I)|_{I=0} = \frac{1}{32(\D+1)} \textrm{Id}_{\D}.
\end{displaymath}
In particular, the frequency map $I \mapsto \n_I \Hg(I)$ is nondegenerate at $I=0$ and hence, by analyticity, nondegenerate on an open dense subset of $\R_{\geq 0}^{\D}$.
\end{cor}



From Theorem \ref{dirtodabnf4} and Corollary \ref{dirtodahessian},
we can conclude the following fact on the application of the classical KAM-theorem (cf. e.g. \cite{kapo, poeschelkam82}) and Nekhoroshev's theorem (see e.g. \cite{lochak, lone, nekh1, nekh2, poeschel0, poeschel}) to the Dirichlet Toda lattice. Consider the set

\begin{cor} \label{kamnekhdirtoda}
For any $\g > 0$, the classical KAM theorem applies to the Dirichlet Toda lattice with Hamiltonian $\HD_{Toda}$ given by (\ref{hamdirtodapq}) and boundary conditions (\ref{dirbdycond}) on an open dense subset of $\R^{2\D}$. Moreover, Nekhoroshev's theorem applies to the Dirichlet Toda lattice on an open neighborhood of the origin in $\R^{2\D}$.
\end{cor}

\begin{remark}
In addition to Corollary \ref{kamnekhdirtoda}, we conjecture that Nekhoroshev's theorem actually holds on the set
  \[ \mathcal{P}^\bullet := \left\{ (q,p) \in \R^{2\D} \big| \, I_k(q,p) > 0 \;\; \forall \, 1 \leq k \leq \D \right\}.
\]
The proof of this fact, along the lines of an analogous proof for the periodic lattice in \cite{ahtk6}, is an ongoing project.
\end{remark}



The remainder of this introduction is devoted to an exposition of our results on the symmetries of the \emph{periodic} lattice needed for the proof of the above results on the normal form and perturbations of the \emph{Dirichlet} lattice. The main tool for the investigation of the periodic lattice are the (noncanonical) coordinates $(b_j,a_j)_{j \in \Z}$ introduced by Flaschka \cite{fla1},
\begin{displaymath}
b_n := -p_n \in \R, \quad a_n := \alpha e^{\frac{1}{2} (q_n - q_{n+1})} \in \R_{>0} \quad (n \in \Z).
\end{displaymath}

In these coordinates the Hamiltonian $H_{Toda}$ takes the simple form
\begin{equation} \label{htodaperflaschka}
H = \frac{1}{2} \sum_{n=1}^N b_n^2 + \sum_{n=1}^N a_n^2,
\end{equation} \label{htodaper}
and the equations of motion are
\begin{equation} \label{flaeqn}
\left\{ \begin{array}{lllll}
 \dot{b}_n & = & a_n^2 - a_{n-1}^2 \\
 \dot{a}_n & = & \frac{1}{2} a_n (b_{n+1} - b_n)
\end{array} \right. \qquad (n \in \Z).
\end{equation}
Note that $(b_{n+N}, a_{n+N}) = (b_n, a_n)$ for any $n \in \Z$, and $\prod_{n=1}^N a_n = \alpha^N$.
Hence we can identify the sequences $(b_n)_{n \in \Z}$ and $(a_n)_{n \in \Z}$ with the vectors $(b_n)_{1 \leq n \leq N} \in \R^N$ and $(a_n)_{1 \leq n \leq N} \in \R_{>0}^N$. In \cite{ahtk1, ahtk2} we studied the normal form of the system of equations (\ref{flaeqn}) on the phase space
\begin{displaymath}
\M := \R^N \times \R_{>0}^N.
\end{displaymath}
This system is Hamiltonian with respect to the nonstandard Poisson structure $J \equiv J_{b,a}$, defined at a point $(b,a) = (b_n, a_n)_{1 \leq n \leq N}$ by
\begin{equation} \label{jdef}
J = \left( \begin{array}{cc}
0 & A \\
-{}^t A & 0 \\
\end{array} \right),
\end{equation}
where $A$ is the $b$-independent $N \times N$-matrix
\begin{equation} \label{adef}
A = \frac{1}{2} \left( \begin{array}{ccccc}
a_1 & 0 & \ldots & 0 & -a_N \\
-a_1 & a_2 & 0 & \ddots & 0 \\
0 & -a_2 & a_3 & \ddots & \vdots \\
\vdots & \ddots & \ddots & \ddots & 0 \\
0 & \ldots & 0 & -a_{N-1} & a_N \\
\end{array} \right).
\end{equation}
The equations (\ref{flaeqn}) can also be written as $\dot{b}_n = \{ b_n, H \}_J$, $\dot{a}_n = \{ a_n, H \}_J$ for $1 \leq n \leq N$, where $\{ \cdot, \cdot \}_J$ is the Poisson bracket corresponding to (\ref{jdef}).


Since the matrix $A$ defined by (\ref{adef}) has rank $N-1$, the Poisson structure $J$ is degenerate.
It admits the two Casimir functions\footnote{A smooth function $C: \M \to \R$ is a Casimir function for $J$ if $\{ C, \cdot \}_J \equiv 0$.}
\begin{equation} \label{casimirdef}
C_1 := -\frac{1}{N} \sum_{n=1}^N b_n \quad \textrm{and} \quad C _2 := \left( \prod_{n=1}^N a_n \right)^\frac{1}{N}.
\end{equation}
Let
\begin{displaymath}
\Mba := \{ (b,a) \in \M : (C_1, C_2) = (\beta, \a) \}
\end{displaymath}
denote the level set of $(C_1, C_2)$ for $(\b, \a) \in \R \times \R_{> 0}$. Note that $(-\b 1_N, \a 1_N) \in \Mba$ where $1_N = (1, \ldots, 1) \in \R^N$. By computing the gradients of $C_1$ and $C_2$, one can show that the sets $\Mba$ are real analytic submanifolds of $\M$ of codimension two. Furthermore the Poisson structure $J$, restricted to $\Mba$, becomes nondegenerate everywhere on $\Mba$ and therefore induces a symplectic structure $\nu_{\b, \a}$ on $\Mba$. In this way, we obtain a symplectic foliation of $\M$ with $\Mba$ being its (symplectic) leaves.

In \cite{ahtk2} we showed that the periodic Toda lattice admits global Birkhoff coordinates. To state this result, we introduce the model space
$$ \P := \R^{2(N-1)} \times \R \times \R_{>0} $$
 endowed with the degenerate Poisson structure $J_0$ whose symplectic leaves are $\R^{2(N-1)} \times \{ \b \} \times \{ \a \}$ endowed with the canonical symplectic structure. More precisely, in \cite{ahtk2} we proved the following result:
\begin{theorem} \label{pertodabirkhoffthm}
There exists a map
\begin{displaymath}
\begin{array}{ccll}
 \Phi: & (\M, J) & \to & (\P, J_0) \\
 & (b,a) & \mapsto & ((x_n, y_n)_{1 \leq n \leq N-1}, C_1, C_2)
\end{array}
\end{displaymath}
with the following properties:
\begin{itemize}
  \item $\Phi$ is a real analytic diffeomorphism.
  \item $\Phi$ is canonical, i.e. it preserves the  Poisson brackets. In particular, the symplectic foliation of $\M$ by $\Mba$ is trivial.
  \item The coordinates $(x_n, y_n)_{1 \leq n \leq N-1}, C_1, C_2$ are global Birkhoff coordinates for the periodic Toda lattice, i.e. the Toda Hamiltonian, when expressed in these coordinates, takes the form 
%
$H \circ \Phi^{-1} = \frac{N\b^2}{2} + H_{\a}(I)$, where $H_{\a}(I)$ is a real analytic function of the action variables
\begin{equation} \label{inperdef}
 I_n := \frac{x_n^2 + y_n^2}{2} \quad (1 \leq n \leq N-1),
\end{equation}
and where $\b$, $\a$ are the values of the Casimirs $C_1$, $C_2$.
\end{itemize}
\end{theorem}


Note that on an open dense subset $\mathcal{M} \setminus D_n$ of the phase space $\mathcal{M}$, where
\begin{equation} \label{dndef}
 D_n := \left\{ (b,a) \in \M | I_n(b,a) = 0 \right\}, \quad (1 \leq n \leq N-1)
\end{equation}
the coordinates $(x_n, y_n)_{1 \leq n \leq N-1}$ are given in terms of action and angle variables $(I_n, \theta_n)_{1 \leq n \leq N-1}$ by
\begin{equation} \label{xnynsqrtinthetan}
 (x_n, y_n) = \sqrt{I_n} (\cos \theta_n, \sin \theta_n),
\end{equation} 
and we have given explicit formulas for these action and angle variables in \cite{ahtk1}.

On the level of the Flaschka variables $(b_j,a_j)_{1 \leq j \leq N}$, the maps $T$ and $S$ introduced in (\ref{rqpdef}) and (\ref{sqpdef}) are given by $\tilde{T}, \tilde{S}: \mathcal{M} \to \mathcal{M}$, with $\tilde{S}(b,a) \equiv (\tilde{S}(b), \tilde{S}(a))$ (analogously for $\tilde{T}$) and
\begin{eqnarray}
        (\tilde{T}(b))_j = b_{j+1}, & (\tilde{T}(a))_j = a_{j+1} & (1 \leq j \leq N), \label{tilderdef} \\
	(\tilde{S}(b))_j = -b_{N-j}, & \qquad (\tilde{S}(a))_j = a_{N-j-1} & (1 \leq j \leq N). \label{tildesdef}
\end{eqnarray}
with the indices in (\ref{tilderdef}) and (\ref{tildesdef}) understood mod $N$. 
Similarly to Fix$(S)$ defined above, we denote by Fix$(\tilde{S})$ the subset of all elements $(b,a) \in \M$ satisfying
\begin{equation} \label{fixsflaschka}
 (b_{N-j}, a_{N-j}) = (-b_j, a_{j+1}) \quad \textrm{for any} \quad 1 \leq j \leq N.
\end{equation} 
In the sequel, we will omit the tilde and write $T$ and $S$ for the operators $\tilde{T}$ and $\tilde{S}$ on $\M$. Note that the Casimirs $C_1$ and $C_2$ defined in (\ref{casimirdef}) are invariant under $T$ and $S$, i.e. for $i=1,2$, $C_i(T(b,a)) = C_i(S(b,a)) = C_i(b,a)$.

In this paper we use the explicit formulas for the action and angle variables proven in \cite{ahtk1} to obtain results on the symmetries of the action and angle variables of the periodic lattice defined by (\ref{xnynsqrtinthetan}).

\begin{theorem} \label{ikthetaktildesba}
 The action-angle variables $(I_n, \theta_n)_{1 \leq n \leq N-1}$ for the periodic Toda lattice given by (\ref{inperdef}) and (\ref{xnynsqrtinthetan}) and the symmetries $T$ and $S$ given in terms of the Flaschka variables $(b,a)$ by (\ref{tilderdef}) and (\ref{tildesdef}) satisfy the following identities. For any $(b,a) \in \M$,
\begin{eqnarray}
 I_n(T(b,a)) & = & I_n(b,a), \label{iktilder} \\
 I_n(S(b,a)) & = & I_{N-n}(b,a). \label{iktildes}
\end{eqnarray}
For any $(b,a) \in \M \setminus D_n$,
\begin{equation} \label{thetaktilder}
 \theta_n(T(b,a)) = \theta_n(b,a) + \frac{2 \pi n}{N} \qquad (\textrm{mod } 2\pi).
\end{equation} 
For any $(b,a) \in \M \setminus D_{N-n}$,
\begin{equation} \label{thetaktildes}
 \theta_n(S(b,a)) = \theta_{N-n}(b,a) + \pi - \frac{4 \pi n}{N} \qquad (\textrm{mod } 2\pi).
\end{equation} 
\end{theorem}

We can also express the symmetry transformations $T$ and $S$ in terms of the Cartesian coordinates $(x_n,y_n)_{1 \leq n \leq N-1}$ of Theorem \ref{pertodabirkhoffthm}, or more conveniently, in terms of the associated complex coordinates $(\zeta_k)_{1 \leq |k| \leq N-1}$, defined for $1 \leq k \leq N-1$ by
\begin{equation} \label{complexdef}
\left\{ \begin{array}{rcl}
\z_k & = & \frac{1}{\sqrt{2}} (x_k - i y_k) \vspace{0.15cm} \\ 
\z_{-k} = \overline{\z_k} & = & \frac{1}{\sqrt{2}} (x_k + i y_k).
\end{array} \right.
\end{equation}
Denote by $\mathcal{Z}$ the linear subspace of $\C^{2N-2}$ consisting of such vectors $(\z_k)_{1 \! \leq \! |k| \! \leq \! N\!-\!1}$. 

Define the map $\sz: \mathcal{Z} \to \mathcal{Z}$, given by
\begin{equation} \label{szdef}
 (\z_k)_{1 \leq |k| \leq N-1} \mapsto (-e^{4 \pi i k / N} \z_{N-k})_{1 \leq |k| \leq N-1}.
\end{equation}
Like the map $\tilde{S}: \mathcal{M} \to \mathcal{M}$, $\sz$ is a linear involution. In fact, the maps $\tilde{S}$ and $\sz$ are conjugate to each other under the coordinate change of Theorem \ref{pertodabirkhoffthm}:

\begin{theorem} \label{xkxnkproof}
In terms of the complex variables $(\z_k)_{1 \leq |k| \leq N-1}$ defined by (\ref{complexdef}) and the Birkhoff map $\Phi$ of Theorem \ref{pertodabirkhoffthm}, the map $S$ is given by $S_\mathcal{Z}$. More precisely, 
\begin{equation} \label{szomega01omega01stilde}
 \sz \circ \Phi = \Phi \circ S.
\end{equation}
\end{theorem}

\emph{Related work:} Similar results as the ones stated in Theorem \ref{ikthetaktildesba} have been obtained for the defocusing nonlinear Schr\"odinger equation, see \cite{grkp}. Besides the aforementioned application to the lattice with Dirichlet boundary conditions, one of the aims of this paper is to show that the technique of expressing symmetry properties of a system in terms of action-angle variables and Birkhoff coordinates can also be applied to a discrete system such as the Toda lattice. The technique of embedding the phase space of a lattice with Dirichlet boundary conditions into the phase space of the corresponding lattice with periodic boundary conditions has already been used in the case of arbitrary FPU chains, see \cite{ahtk5, rink06}. The Toda lattice is a special case of an FPU chain, however an especially interesting one due to its integrability properties.

\emph{Outline of the paper:} In section \ref{tools} we review the Lax pair of the periodic Toda lattice and some auxiliary results on the spectrum of the Jacobi matrix $L(b,a)$ associated to an element $(b,a) \in \M$. In particular, we express the symmmetries $T$ and $S$ in terms of the (periodic and Dirichlet) spectrum of $L(b,a)$. Moreover, we discuss some important properties of the Riemann surface $\Sigma_{b,a}$ associated to $(b,a)$ (section \ref{symmriemann}). In section \ref{symmactang}, we use the results of sections \ref{tools} and \ref{symmriemann} to prove Theorems \ref{ikthetaktildesba} and \ref{xkxnkproof}. 
In section \ref{dirbirkhoff}, we consider the lattice with Dirichlet boundary conditions and prove a result similar to Theorem \ref{pertodabirkhoffthm}, on the basis of which we prove Theorem \ref{dirtodabnf4} 
in section \ref{dirkamnekh}. 

\section{Symmetries and spectra} \label{tools}

It is well known (cf. e.g. {\cite{toda}}) that the system (\ref{flaeqn})
admits a Lax pair formulation $\dot{L} = \frac{\partial L}{\partial t} = [B,
L]$, where $L \equiv L^+ (b, a)$ is the periodic Jacobi matrix defined by
\begin{equation}
  \label{jacobi} L^{\pm} (b, a) : = \left( \begin{array}{ccccc}
    b_1 & a_1 & 0 & \ldots & \pm a_N\\
    a_1 & b_2 & a_2 & \ddots & \vdots\\
    0 & a_2 & b_3 & \ddots & 0\\
    \vdots & \ddots & \ddots & \ddots & a_{N - 1}\\
    \pm a_N & \ldots & 0 & a_{N - 1} & b_N
  \end{array} \right),
\end{equation}
and a skew-symmetric matrix $B$ given in {\cite{ahtk1}}. Hence the flow of
$\dot{L} = [B, L]$ is isospectral.

Let us now collect a few results from {\cite{moer}} and {\cite{toda}} of the
spectral theory of Jacobi matrices needed in the sequel. Denote by $\M^{\C}$ the
complexification of the phase space $\M$,
\[ \M^{\C} =\{(b, a) \in \C^{2 N} : \textrm{Re } a_j > 0 \hspace{1em} \forall
   \, 1 \leq j \leq N\}. \]
For $(b, a) \in \M^{\C}$ we consider for any complex number ${\l}$ the difference
equation
\begin{equation}
  \label{diff} (R_{b, a} y) (k) = {\l}y (k) \hspace{1em} (k \in \Z)
\end{equation}
where $y (\cdot) = y (k)_{k \in \Z} \in \C^{\Z}$ and $R_{b, a}$ is the difference
operator
\begin{equation}
  \label{rdef} R_{b, a} = a_{k - 1} T^{- 1} + b_k T^0 + a_k T^1
\end{equation}
with $T^m$ denoting the shift operator of order $m \in \Z$, in accordance with (\ref{rqpdef}).

\emph{Fundamental solutions:} The two fundamental solutions $y_1 (\cdot, {\l})$ and $y_2 (\cdot, {\l})$ of (\ref{diff}) are defined by the standard initial conditions $y_1 (0, {\l}) = 1$, $y_1 (1, {\l}) = 0$ and $y_2 (0, {\l}) = 0$, $y_2 (1, {\l}) = 1$. 
For each $k \in \Z$, $y_i (k, {\l},
b, a)$, $i = 1, 2$, is a polynomial in ${\l}$ of degree at most $k - 1$ and
depends real analytically on $(b, a)$ (see {\cite{moer}}). In particular, one
easily verifies that $y_2 (N + 1, {\l}, b, a)$ is a polynomial in ${\l}$ of
degree $N$ with leading term $\a^{-N} {\l}^N$, whereas $y_1 (N, {\l})$ is a
polynomial in ${\l}$ of degree less than $N$.

\emph{Discriminant:} We denote by $\Delta ({\l}) \equiv \Delta ({\l}, b,
a)$ the \emph{discriminant} of (\ref{diff}), defined by
\begin{equation}
  \label{discrdef} \Delta ({\l}) : = y_1 (N, {\l}) + y_2 (N + 1, {\l}) .
\end{equation}
In the sequel, we will often write $\Delta_{{\l}}$ for $\Delta ({\l})$. As
$y_2 (N + 1, {\l}) = \a^{- N} {\l}^N + \ldots$ and $y_1 (N, {\l}) = O ({\l}^{N -
1})$, $\Delta ({\l}, b, a)$ is a polynomial in ${\l}$ of degree $N$ with
leading term $\a^{- N} {\l}^N$, and it depends real analytically on $(b, a)$
(see e.g. {\cite{toda}}). According to Floquet's Theorem (see e.g.
{\cite{teschl}}), for ${\l} \in \R$ given, (\ref{diff}) admits a periodic or
antiperiodic solution of period $N$ if the discriminant $\Delta ({\l})$
satisfies $\Delta ({\l}) = + 2$ or $\Delta ({\l}) = - 2$, respectively. These
solutions correspond to eigenvectors of $L^+$ or $L^-$, respectively, with
$L^{\pm}$ defined by (\ref{jacobi}). In particular,
\begin{equation}
  \label{deltapm2repr} \Delta ({\l}) - 2 = \a^{- N} \prod_{j = 1}^N ({\l} -
  {\l}_j^+), \hspace{1em} \Delta ({\l}) + 2 = \a^{- N} \prod_{j = 1}^N ({\l} -
  {\l}_j^-) .
\end{equation}
It turns out to be more convenient to combine these two cases by considering
the periodic Jacobi matrix $Q \equiv Q(b,a) = L((b,b),(a,a))$ of size $2 N$ defined by
\[ Q = \left( \begin{array}{cccc|cccc}
     b_1 & a_1 & \ldots & 0 & 0 & \ldots & 0 & a_N\\
     a_1 & b_2 & \ddots & \vdots & 0 & \ldots &  & 0\\
     \vdots & \ddots & \ddots & a_{N - 1} & \vdots &  &  & \vdots\\
     0 & \ddots & \hspace{0.75em} \hspace{0.75em} a_{N - 1} & b_N & a_N &
     \ldots & 0 & 0\\
     \hline
     0 & \ldots & 0 & a_N & b_1 & a_1 & \ldots & 0\\
     0 & \ldots &  & 0 & a_1 & b_2 & \ddots & \vdots\\
     \vdots &  &  & \vdots & \vdots & \ddots & \ddots & a_{N - 1}\\
     a_N & \ldots & 0 & 0 & 0 & \ddots & \hspace{0.75em} \hspace{0.75em} a_{N
     - 1} & b_N
   \end{array} \right) . \]
Then the spectrum of the matrix $Q$ is the union of the spectra of the
matrices $L^+$ and $L^-$ and therefore the zero set of the polynomial
$\Delta^2_{{\l}} - 4$. The function $\Delta^2_{{\l}} - 4$ is a polynomial in
${\l}$ of degree $2 N$ and admits a product representation
\begin{equation} \label{delta2lrepr} 
  \Delta^2_{{\l}} - 4 = \a^{- 2 N} \prod_{j = 1}^{2N} ({\l} - {\l}_j) .
\end{equation}
The factor $\a^{- 2 N}$ in (\ref{delta2lrepr}) comes from the above mentioned
fact that the leading term of $\Delta ({\l})$ is $\a^{- N} {\l}^N$.

For any $(b, a) \in \M$, the matrix $Q$ is symmetric and hence the eigenvalues
$({\l}_j)_{1 \leq j \leq 2 N}$ of $Q$ are real. When listed in increasing
order and with their algebraic multiplicities, they satisfy the following
relations (cf. {\cite{moer}})
\begin{equation} \label{eigenvorder}
 {\l}_1 < {\l}_2 \leq {\l}_3 < {\l}_4 \leq {\l}_5 < \ldots {\l}_{2 N - 2}
   \leq {\l}_{2 N - 1} < {\l}_{2 N}.
\end{equation} 
As explained above, the ${\l}_j$ are periodic or antiperiodic eigenvalues of
$L$ and thus eigenvalues of $L^+$ or $L^-$ according to whether $\Delta
({\l}_j) = 2$ or $\Delta ({\l}_j) = - 2$. One has (cf. {\cite{moer}})
\begin{equation} \label{deltalambdapm2} 
  \Delta ({\l}_1) = (- 1)^N \cdot 2, \hspace{1em}
  \Delta ({\l}_{2n}) = \Delta ({\l}_{2n+1}) = (-1)^{n+N} \cdot 2,
  \hspace{1em} \Delta ({\l}_{2N}) = 2.
\end{equation}

The open intervals $({\l}_{2 n}, {\l}_{2 n + 1})$ are referred to
as the \emph{$n$-th spectral gap} and $\gamma_n : = {\l}_{2 n + 1} -
{\l}_{2 n}$ as the \emph{$n$-th gap length}. Note that $| \Delta ({\l}) |
> 2$ on the spectral gaps. We say that the $n$-th gap is \emph{open} if
$\gamma_n > 0$ and \emph{collapsed} otherwise.

Generalizing (\ref{diff}), we consider for any $1 \leq k \leq N$ solutions $y(\cdot |k)$
of the equation $(R_{T^k(b,a)}y)(n) = \l y(n)$, i.e. the equation
\begin{equation} \label{ynkdef}
   b_{n + k} y (n|k) + a_{n + k} y (n + 1| k) + a_{n + k - 1} y(n - 1| k) = {\l}y (n|k) .
\end{equation}
The fundamental solutions $y_1 (\cdot, {\l}|k)$ and $y_2 (\cdot, {\l}|k)$ are
solutions of (\ref{ynkdef}) with the same standard initial conditions $y_1 (0,
{\l}|k) = 1$, $y_1 (1, {\l}|k) = 0$ and $y_2 (0, {\l}|k) = 0$, $y_2 (1,
{\l}|k) = 1$, as in the case $k=0$. The following lemma on the (anti)periodic eigenvalues $\l_j(k) \equiv \l_j(T^k(b,a))$ and the discriminant of $T(b,a)$,
\begin{displaymath}
 \Delta ({\l}|k) := \Delta_\l(T^k(b,a)) = y_1 (N|k) + y_2 (N + 1| k)
\end{displaymath}
is shown in (\cite{toda}, section 4.2).

\begin{lemma} \label{deltalambdatranslemma}
 The discriminant is invariant under the translation map $T$, i.e. for any $\l \in \R$,
\[ \Delta ({\l}|k) = \Delta ({\l}).
\]
Consequently, for any $1 \leq j \leq 2 N$,
\[ {\l}_j (k) = {\l}_j (0) = {\l}_j.
\]
\end{lemma}


Let us now describe the behavior of these quantities under the symmetry transformation $S$.
\begin{lemma} \label{symmetrystildebalemma}
  Let ${\l} \in$ spec $Q(b,a)$. Then $-\l \in$ spec $Q(S(b, a))$.
\end{lemma}

\begin{proof}
  If ${\l} \in$ spec $Q(b, a)$, there exists $v(n)_{n \in \Z}$ such that $v(n + N) = \pm v(n)$ for any $n \in \Z$ and $L(b, a) v = {\l}v$, i.e. $v$ is an
  $N$-periodic or antiperiodic eigenvector of $L^\pm(b, a)$. Then $w (n)_{n \in \Z}$, defined by $w(k) : = (-1)^{k + 1} v(k)$, is an $N$-periodic or
  antiperiodic eigenvector of $L^\pm(S(b, a))$ to the eigenvalue $-\l$, which can be checked by a direct computation.
\end{proof}

\begin{lemma} \label{ljstildelemma}
Let $(b, a) \in \M$. Then for any $1 \leq j \leq 2 N$,
  \begin{equation} \label{ljtildesba1} 
{\l}_j (S(b, a)) = -{\l}_{2N+1-j}(b,a).
  \end{equation}
  If $N$ is even, for any $1 \leq j \leq N$,
  \begin{equation}
    \label{ljtildesba2} {\l}_j^+ (S(b, a))=-{\l}_{N+1-j}^+(b,a), \hspace{0.75em} {\l}_j^-(S(b,a))=-{\l}_{N+1-j}^-(b,a),
  \end{equation}
  and if $N$ is odd, for any $1 \leq j \leq N$,
  \begin{equation}
    \label{ljtildesba3} {\l}_j^+ (S(b, a)) = - {\l}_{N+1-j}^-(b,a), \hspace{0.75em} {\l}_j^-(S(b,a))=-{\l}_{N+1-j}^+ (b,a).
  \end{equation}
Moreover, for any $\1N1$,
  \begin{equation} \label{gammaksba} 
    \gamma_k(S(b,a)) = \gamma_{N-k}(b,a).
  \end{equation}
\end{lemma}

\begin{proof}
  It follows from Lemma \ref{symmetrystildebalemma} that
  \[ \{{\l}_j (S(b, a)) |1 \leq j \leq 2 N\}=\{- {\l}_j (b, a)| 1 \leq j \leq 2 N\}. \]
  Because of the ordering of the eigenvalues described in (\ref{eigenvorder}), (\ref{ljtildesba1}) follows. The
  statements (\ref{ljtildesba2}) and (\ref{ljtildesba3}) follow from the ordering of periodic and anti-periodic eigenvalues described in (\ref{deltalambdapm2}). The statement (\ref{gammaksba}) on the gap lengths follows from (\ref{ljtildesba1}).
\end{proof}

\begin{cor} \label{deltalstildecor}
 For any $(b,a) \in \M$,
\begin{equation} \label{deltaltildesbaformel}
 \Delta_\l(S(b,a)) \mp 2 
= (-1)^N \Delta_{-\l}(b,a) \mp 2,
\end{equation}
and hence, with $\dot{\Delta}_\l = \partial_{\l} \Delta_\l$,
\begin{equation} \label{deltaltildesbaformel1}
 \dot{\Delta}_\l(S(b,a)) = (-1)^{N+1}\dot{\Delta}_{-\l}(b,a)
\end{equation}
as well as
\begin{equation} \label{deltaltildesbaformel2}
 \Delta_\l^2(S(b,a)) - 4 = \Delta_{-\l}^2(b,a) - 4.
\end{equation}
\end{cor}

\begin{proof}
Since (\ref{deltaltildesbaformel1}) and (\ref{deltaltildesbaformel2}) immediately follow from (\ref{deltaltildesbaformel}), it suffices to prove the latter formula. By (\ref{deltapm2repr}), we have
  \[ \Delta_{{\l}} (S(b,a)) \mp 2 = \a^{-N} \prod_{j = 1}^N ({\l} - {\l}_j^\pm (S(b,a)). \]
If $N$ is even, it follows from (\ref{ljtildesba2}) that
  \begin{eqnarray*}
    \Delta_{\l} (S(b, a)) \mp 2 & = & \a^{-N} \prod_{j = 1}^N ({\l} + {\l}_{N +1-j}^\pm (b,a))\\
    & = & \a^{-N} \prod_{j = 1}^N ({\l} + {\l}_j^\pm (b, a))\\
    & = & \a^{-N} \prod_{j=1}^N \left( (-1) ((-{\l}) - {\l}_j^\pm (b, a)) \right) \\
& = & (-1)^N (\Delta_{-\l}(b,a) \mp 2).
  \end{eqnarray*}
  This proves (\ref{deltaltildesbaformel}) for even $N$. For odd $N$, one uses (\ref{ljtildesba3}) instead of (\ref{ljtildesba2}) to obtain (\ref{deltaltildesbaformel}) in an analogous computation.
\end{proof}

\emph{Dirichlet eigenvalues:} For $(b, a) \in \M$, the set Iso$(b, a)$ of
all elements $(b', a') \in $ so that $Q (b', a')$ has the same spectrum as $Q
(b, a)$ is described with the help of the Dirichlet eigenvalues $\mu_1 < \mu_2
< \ldots < \mu_{N - 1}$ of (\ref{diff}) defined by
\begin{equation}
  \label{mundef} y_1 (N + 1, \mu_n) = 0.
\end{equation}
They coincide with the eigenvalues of the $(N - 1) \times (N - 1)$-matrix $L_2
= L_2 (b, a)$ given by
\begin{equation}
  \label{l2bamatrixdef} \left( \begin{array}{ccccc}
    b_2 & a_2 & 0 & \ldots & 0\\
    a_2 & b_3 & \ddots & \ddots & \vdots\\
    0 & \ddots & \ddots & \ddots & 0\\
    \vdots & \ddots & \ddots & \ddots & a_{N - 1}\\
    0 & \ldots & 0 & a_{N - 1} & b_N
  \end{array} \right) .
\end{equation}
In the sequel, we will also refer to $\mu_1, \ldots, \mu_{N - 1}$ as the Dirichlet eigenvalues of $L(b,a)$. It is shown in {\cite{ahtk1}} that everywhere in the real phase space $\M$ and for any $1 \leq n \leq N-1$, ${\l}_{2 n} \leq \mu_n \leq {\l}_{2n + 1}$.

We now consider the behaviour of the $\mu_n$'s under the transformations $T^k$ and $S$ for any $\1N1$, i.e. $\mu_n(T^k(b,a)))$ and $\mu_n(S^k(b,a)))$; we write $\mu_n(k) = \mu_n(T^k(b,a)))$ and $\mu_n(S) = \mu_n(S(b,a)))$ in the sequel. By the definition (\ref{mundef}), these quantities satisfy are defined by
\begin{displaymath}
 y_1(N+1,\mu_n(k)|k) = 0, \qquad y_1(N+1,\mu_n(S); S(b,a)) = 0.
\end{displaymath}
As in the general case, for any $1 \leq n,k \leq N-1$, $\l_{2n}(k) \leq \mu_n(k) \leq \l_{2n+1}(k)$, 
and analogously for $\mu_n(S)$. Therefore, by Lemma \ref{deltalambdatranslemma}, $$ {\l}_{2n} \leq \mu_n(k) \leq {\l}_{2n+1}. $$ In general however, $\mu_n \neq \mu_n(k)$ and $\mu_n \neq \mu_n(S)$. The following lemma gives a connection between the $\mu_n(k)$'s and the $\mu_n(S)$'s.


\begin{lemma}
  \label{munstildelemma} For any $1 \leq n \leq N - 1$,
\begin{equation} \label{munsmunnmn}
 \mu_n(S) = - \mu_{N-n}(N-2).
\end{equation} 
\end{lemma}

\begin{proof}
As mentioned above, the $\mu_j$'s coincide with the eigenvalues of the matrix $L_2 = L_2 (b, a)$ given by (\ref{l2bamatrixdef}).
 It follows that the $\mu_j(S)$'s are the eigenvalues of the matrix $\tilde{L}_2 = L_2(S(b,a))$, given by (recall the formula (\ref{tildesdef}) for $S(b,a)$)
  \[ \tilde{L}_2 = \left( \begin{array}{ccccc}
       -b_{N - 2} & a_{N - 3} & 0 & \ldots & 0\\
       a_{N - 3} & \ddots & \ddots & \ddots & \vdots\\
       0 & \ddots & \ddots & a_1 & 0\\
       \vdots & \ddots & a_1 & -b_1 & a_N\\
       0 & \ldots & 0 & a_N & -b_N
     \end{array} \right) . \]
If a real number $\mu$ is an eigenvalue of $\tilde{L}_2$, then $-\mu$ is an eigenvalue of the related matrix
  \begin{equation} \label{l2matrixbis}
   \left( \begin{array}{ccccc}
       b_{N - 2} & a_{N - 3} & 0 & \ldots & 0\\
       a_{N - 3} & b_{N - 3} & \ddots & \ddots & \vdots\\
       0 & \ddots & \ddots & a_1 & 0\\
       \vdots & \ddots & a_1 & b_1 & a_N\\
       0 & \ldots & 0 & a_N & b_N
     \end{array} \right),
  \end{equation}
which can be shown by explicitly constructing appropriate eigenvectors. The matrix (\ref{l2matrixbis}) however is identical to the matrix $L_2 \left( T^{N-2}(b,a) \right)$. This proves
 (\ref{munsmunnmn}) and therefore Lemma \ref{munstildelemma}.
\end{proof}

\section{Symmetries and Riemann surfaces} \label{symmriemann}


Denote by $\Sigma_{b, a}$ the Riemann surface obtained as the compactification of the affine curve
\begin{equation} \label{algcurve}
 \mathcal{C}_{b, a} := \{({\l}, z) \in \C^2 : z^2 = R(\l) \}
\end{equation}
for
\begin{equation} \label{rlriemanndef}
 R(\l) := \Delta^2_{{\l}} (b,a) - 4 = \a^{-2N} \prod_{j=1}^{2N} (\l - \l_j(b,a))
\end{equation}
by (\ref{delta2lrepr}). Note that $\mathcal{C}_{b, a}$ is a two-sheeted curve with the ramification points $({\l}_i, 0)_{1 \leq i \leq 2 N}$, identified with ${\l}_i$ in the sequel, and that $\mathcal{C}_{b, a}$ and $\Sigma_{b, a}$ are spectral invariants; the Riemann surface $\Sigma_{b, a}$ is obtained from $\mathcal{C}_{b, a}$ by adding two (unramified) points at infinity, $\infty^+$ and $\infty^-$, one on each of the two sheet, i.e.
\begin{displaymath}
 \Sigma_{b, a} := \mathcal{C}_{b, a} \cup \{ \infty^+, \infty^- \}.
\end{displaymath}

Strictly speaking, $\Sigma_{b, a}$ is a Riemann surface only if the spectrum of $Q(b, a)$ is simple. 
If the spectrum of $Q (b, a)$ is \emph{not} simple, $\Sigma (b, a)$ becomes a Riemann surface after doubling the multiple eigenvalues - see e.g. section $2$ of {\cite{kato}}. We will encounter this case in the discussion of Fix$(S)$ in section \ref{fixedpointset}. For the rest of this section, we however assume that $(b,a) \in \M^\bullet$, where
\begin{equation} \label{mbulletdef}
 \M^\bullet = \M \setminus \bigcup_{n=1}^{N-1} D_n
\end{equation}
for the sets $(D_n)_{1 \leq n \leq N-1}$ defined by (\ref{dndef}). Since $I_n(b,a) = 0$ iff $\l_{2n}(b,a) = \l_{2n+1}(b,a)$ (cf. \cite{ahtk1} or the formula (\ref{ikbaformel}) for $I_n$ cited below), $(b,a) \in \M^\bullet$ iff all gaps in the spectrum of $Q(b,a)$ are open, or equivalently, if the spectrum of $Q(b,a)$ is simple.

\begin{lemma} \label{psifnstildecor}
Let $(b, a) \in \M^\bullet$. Then
\begin{eqnarray*}
 \textrm{(i)} \quad \Sigma_{T(b, a)} & = & \Sigma_{b, a}, \\
\textrm{(ii)} \quad \Sigma_{S(b, a)} & \cong & \Sigma_{b, a}.
\end{eqnarray*}
\end{lemma}

\begin{proof}
 (i) This follows from the invariance of $\Delta_\l^2(b,a) - 4$ under $T$ (see Lemma \ref{deltalambdatranslemma}).

(ii) By Corollary \ref{deltalstildecor}, in particular equation (\ref{deltaltildesbaformel2}), $(\l,z) \in \Sigma_{S(b, a)}$ if and only if $(-\l,z) \in \Sigma_{b, a}$. This proves the claim.
\end{proof}

\emph{Dirichlet divisors:} To the Dirichlet eigenvalues $(\mu_n)_{1 \leq n \leq N - 1}$ we associate the points $(\mu_n^{\ast})_{1 \leq n
\leq N - 1}$ on the surface $\Sigma_{b,a}$,
\begin{equation} \label{munstarred}
  \mu_n^{\ast} : = \left( \mu_n, \sqrt[\ast]{\Delta^2_{\mu_n} - 4} \right) \hspace{0.75em} \textrm{with}
  \hspace{0.75em} \hspace{0.75em} \sqrt[\ast]{\Delta^2_{\mu_n} - 4} : = y_1(N, \mu_n) - y_2 (N + 1, \mu_n),
\end{equation}
where we used that $\Delta^2_{\mu_n} - 4 = \left( y_1 (N, \mu_n) - y_2 (N + 1, \mu_n) \right)^2$.

\begin{lemma} \label{sqrtstarreddeltamuklemma}
For any $(b,a) \in \M^\bullet$ and any $\1N1$,
\begin{equation} \label{sqrtstarreddeltamuk}
 \sqrt[*]{\Delta^2_{\mu_k} - 4} = \frac{1}{2 \a^N} \; \dot{\mu}_k \!\! \prod_{1 \leq l \leq N-1 \atop l \neq k} (\mu_k - \mu_l).
\end{equation} 
\end{lemma}

\begin{proof}
 By the definition (\ref{munstarred}) of the starred square root, we have to consider $y_1(N,\mu_k) - y_2(N+1,\mu_k)$. 
One can show (see \cite{toda}, section 4.6) that (with the notation $\cdot = \frac{\partial}{\partial t}$)
\begin{equation} \label{sqrtstarred2}
 y_1(N,\mu_k) - y_2(N+1,\mu_k) = \frac{1}{2 a_N} \dot{y}_1(N+1, \mu_k).
\end{equation}
To evaluate the right hand side of the last identity, again following \cite{toda}, first note that we can write
\begin{equation} \label{y1nplus1polynom}
 y_1(N+1,\l) = -\frac{a_N}{A} \prod_{j=1}^{N-1} (\l - \mu_j),
\end{equation}
The identity (\ref{y1nplus1polynom}) follows from $y_1(n,\l) = -\frac{a_N}{\a^N} (\l^{n-2} - O(\l^{n-3}))$ and the defining property $y_1(N+1,\mu_k) = 0$. Differentiating (\ref{y1nplus1polynom}) with respect to time and evaluating at $\l = \mu_k$, we obtain
\begin{equation} \label{sqrtstarred3}
 \dot{y}_1(N+1, \mu_k) = \frac{a_N}{A} \; \dot{\mu}_k \!\! \prod_{1 \leq l \leq N-1 \atop l \neq k} (\mu_k - \mu_l).
\end{equation}
Combining the definition (\ref{munstarred}), (\ref{sqrtstarred2}), and (\ref{sqrtstarred3}), we obtain (\ref{sqrtstarreddeltamuk}).
\end{proof}

It follows from Lemma \ref{sqrtstarreddeltamuklemma} that
\begin{displaymath} 
 \sqrt[*]{\Delta^2_{\mu_k(S)}(S) - 4} = \frac{1}{2 \a^N} \; \dot{\mu}_k(S) \!\! \prod_{1 \leq l \leq N-1 \atop l \neq k} (\mu_k(S) - \mu_l(S)).
\end{displaymath}
Consequently, by Lemma \ref{munstildelemma},
\setlength\arraycolsep{1.5pt}{
\begin{eqnarray*}
 \sqrt[*]{\Delta^2_{\mu_k(S)} (S) - 4} \! & = & -\frac{\dot{\mu}_{N-k}(N\!-\!2)}{2 \a^{N}} \!\!\!\! \prod_{1 \leq l \leq N-1 \atop l \neq k} \!\!\!\! (-\mu_{N-k}(N\!-\!2) + \mu_{N-l}(N\!-\!2)) \\
& = & \!\! \frac{(-1)^{N\!+\!1}}{2 \a^{N}} \dot{\mu}_{N\!-\!k}(N\!-\!2) \!\!\!\!\!\! \prod_{1 \leq l \leq N-1 \atop l \neq N-k} \!\!\!\!\!\! (\mu_{N-k}(N\!-\!2) \!-\! \mu_l(N\!-\!2)). 
\end{eqnarray*}}
On the other hand, it follows directly from (\ref{sqrtstarreddeltamuk}) that
\begin{displaymath} 
 \sqrt[*]{\Delta_{\mu_{N-k}(N-2)}^2(N-2)-4} = \frac{\dot{\mu}_{N-k}(N\!-\!2)}{2 \a^N} \!\!\!\! \prod_{1 \leq l \leq N-1 \atop l \neq N-k} \!\!\!\! (\mu_{N\!-\!k}(N\!-\!2) - \mu_l(N\!-\!2)).
\end{displaymath}
Hence we have the identity
\begin{equation} \label{vorzeichen3formel}
 \sqrt[*]{\Delta^2_{\mu_k(S)}(S) - 4} = (-1)^{N+1} \cdot \sqrt[*]{\Delta_{\mu_{N-k}(N-2)}^2(N-2)-4},
\end{equation}
which we will use for the computation of $\theta_n(S(b,a))$ (see section \ref{symmactang}).


\emph{Canonical sheet and canonical root:} For $(b,a) \in \M^\bullet$ the canonical sheet of $\Sigma_{b, a}$ is given by the set of points $({\l}, \sqrt[c]{\Delta_{{\l}}^2 - 4})$ in $\mathcal{C}_{b, a}$, where the $c$-root $\sqrt[c]{\Delta_{{\l}}^2 - 4}$ is defined on $\C \setminus \bigcup_{n = 0}^N ({\l}_{2 n}, {\l}_{2 n + 1})$ (with ${\l}_0 : = - \infty$ and ${\l}_{2 N + 1} := \infty$) and determined by the sign condition
\begin{equation}
  \label{croot} - i \sqrt[c]{\Delta_{{\l}}^2 - 4} > 0 \hspace{1em}
  \textrm{for} \hspace{1em} {\l}_{2 N - 1} < {\l} < {\l}_{2 N} .
\end{equation}
As a consequence one has for any $1 \leq n \leq N$
\begin{equation} \label{croot2} 
\textrm{sign} \hspace{0.75em} \sqrt[c]{\Delta_{{\l} - i
  0}^2 - 4} = (- 1)^{N + n - 1} \hspace{1em} \textrm{for} \hspace{1em}
  {\l}_{2 n} < {\l} < {\l}_{2 n + 1} .
\end{equation}

\emph{Cycles on $\Sigma_{b,a}$:} It is convenient to introduce the projection $\pi \equiv \pi_{b,a}: \mathcal{C}_{b,a} \to \C$ onto the $\l$-plane, i.e. $\pi_{b,a}(\l,w) = \l$ and its extension to a map $\pi_{b,a}: \Sigma_{b,a} \to \C \cup \{ \infty \}$, where $\pi_{b,a}(\infty^\pm) = \infty$.

We now introduce the cycles $(c_k)_{\1N1}$ and $(d_k)_{\1N1}$ on $\Sigma_{b,a}$. Denote by $(c_k)_{\1N1}$ the cycles on the canonical sheet of $\mathcal{C}_{b,a}$ so that $\pi(c_k)$ is a counterclockwise oriented closed curve in $\C$, containing in its interior the two ramification points $\l_{2k}$ and $\l_{2k+1}$, whereas all other ramification points are outside of $\pi(c_k)$.

Denote by $(d_k)_{\1N1}$ pairwise disjoint cycles on $\mathcal{C}_{b,a} \setminus \{ (\l_k)_{\1N1} \}$ so that for any $1 \leq n,k \leq N-1$, the intersection indices with the cycles $(c_n)_{1 \leq n \leq N-1}$ with respect to the orientation on $\Sigma_{b,a}$, induced by the complex structure, are given by $c_n \circ d_k = \delta_{nk}$. In order to be more precise, choose
the cycles $d_k$ in such a way that (i) the projection $\pi_{b,a}(d_k)$ of $d_k$ is a smooth,
 convex counterclockwise oriented curve in $\C \setminus ((\l_1, \l_{2k}) \cup (\l_{2k+1}, \infty))$ and (ii) the points of $d_k$ whose projection by $\pi_{b,a}$ onto the $\l$-plane have a negative imaginary part lie on the canonical sheet of $\Sigma_{b,a}$.

\emph{Abelian differentials:} Let $(b, a) \in \M^\bullet$ and $1 \leq n \leq N - 1$. Then there exists a unique polynomial $\psi_n ({\l})$ of degree at most $N - 2$ such that for any $1 \leq k \leq N - 1$
\begin{equation} \label{psi}
 \frac{1}{2 \pi} \int_{c_k} \frac{\psi_n({\l})}{\sqrt{\Delta^2_{{\l}} - 4}} \hspace{0.25em} d{\l} = \delta_{kn} .
\end{equation}
Using the definition of the cycles $c_k$ given above, we can rewrite (\ref{psi}) as
\begin{equation} \label{psi2}
 \frac{1}{\pi} \int_{\l_{2k}}^{\l_{2k+1}} \frac{\psi_n({\l})}{\sqrt[c]{\Delta^2_{\l - i0} - 4}} \, d{\l} = \delta_{kn} .
\end{equation}
The coefficients of the polynomials $\psi_n(\l)$ can be computed explicitly, see e.g. Appendix A of \cite{teschl}.

Alternatively, we can rewrite (\ref{psi}) such that the differentials $(\eta_n)_{1 \leq n \leq N-1}$ on $\Sigma_{b,a}$ defined by
\begin{equation} \label{omegalformel}
 \eta_n = \frac{1}{2\pi} \frac{\psi_n(\l)}{\sqrt{\Delta_\l^2 - 4}} \, d\l
\end{equation}
satisfy the conditions
\begin{displaymath}
 \int_{c_k} \eta_n = \delta_{kn} \quad \forall \; 1 \leq k,n \leq N-1
\end{displaymath}
and are therefore a base of normalized Abelian differentials of the first kind on $\Sigma_{b,a}$. For the $d_l$-periods of the differentials $(\eta_k)_{\1N1}$, we write
\begin{equation} \label{taujldef}
  \tau_{lk} = \int_{d_l} \eta_k.
\end{equation}
Note that (see also Appendix A of \cite{teschl}) for any $1 \leq l,k \leq N-1$ and any $(b,a) \in \M^\bullet$,
\begin{equation} \label{taujkiniR}
 \tau_{lk} \in i\R.
\end{equation}


\begin{lemma} \label{psinlambdatildetslemma}
Let $(b, a) \in \M^\bullet$. Then for any real ${\l}$,
\begin{eqnarray*}
 \textrm{(i)} \quad \psi_n ({\l}) (T(b, a)) & = & \psi_n(\l) (b, a), \\
\textrm{(ii)} \quad \psi_n ({\l}) (S(b, a)) & = & (- 1)^N \psi_{N - n} (- {\l}) (b, a). 
\end{eqnarray*}
\end{lemma}

\begin{proof}
The statement (i) follows from the invariance of the $\l_j$'s, the discriminant $\Delta_\l$, and the surface $\Sigma_{b,a}$ under the translation map $T$, as stated in Lemma \ref{deltalambdatranslemma}. It thus remains to prove (ii). By (\ref{psi2}) and the sign condition (\ref{croot2}), 
$\psi_n({\l})(S(b,a))$ must satisfy the identity
  \begin{equation} \label{psinlsigmababedingung} 
\frac{1}{\pi} \int_{{\l}_{2 k} (S(b, a))}^{{\l}_{2 k + 1} (S(b, a))} \frac{\psi_n ({\l}) (
    S(b, a))}{\sqrt[+]{\Delta^2_{{\l}} (S(b, a)) - 4}} \hspace{0.25em} d{\l} = (- 1)^{N + k + 1} \delta_{kn} .
  \end{equation}
  Substituting the claimed formula for $\psi_n ({\l}) (S(b, a))$
  into the left hand side of (\ref{psinlsigmababedingung}), it follows from Lemma \ref{ljstildelemma} and Corollary \ref{deltalstildecor}, in particular equations (\ref{ljtildesba1}) and (\ref{deltaltildesbaformel2}), that
  \[ \frac{1}{\pi} \int_{{\l}_{2 k} (S)}^{{\l}_{2 k + 1} (S)} \frac{(- 1)^N \psi_{N - n} (-\l)}{\sqrt[+]{\Delta_{{\l}}^2 (S) - 4}} \hspace{0.25em} d{\l} = \frac{(- 1)^N}{\pi} \int_{- {\l}_{2 N + 1 - 2 k}}^{- {\l}_{2 N - 2 k}} \frac{\psi_{N - n} (- {\l})}{\sqrt[+]{\Delta_{- {\l}}^2 - 4}} \hspace{0.25em} d{\l}. \]
  By making the substitution ${\l} \mapsto - {\l}$ and reversing the integration direction in the integral on the right hand side of the last
  identity, we obtain
  \[ \frac{1}{\pi} \int_{{\l}_{2 k} (S)}^{{\l}_{2 k + 1} (S)} \frac{(- 1)^N \psi_{N - n} (-{\l})}{\sqrt[+]{\Delta_{{\l}}^2 (S) - 4}} \hspace{0.25em}
     d{\l} = \frac{(- 1)^N}{\pi} \int_{{\l}_{2 (N - k)}}^{{\l}_{2 (N - k) + 1}} \hspace{-0.25em} \frac{\psi_{N - n} ({\l})}{\sqrt[+]{\Delta_{{\l}}^2 -
     4}} \hspace{0.25em} d{\l} \]
  Again by (\ref{psi2}) and (\ref{croot2}), the integral on the right hand side of the last equation equals $\pi \cdot (-1)^{N+(N-k)+1} \delta_{N-k,N-n}$, hence we obtain
\begin{eqnarray}
 \frac{1}{\pi} \int_{{\l}_{2k} (S(b, a))}^{{\l}_{2 k + 1} (S(b, a))} \frac{(- 1)^N \psi_{N - n}(-\l)(b,a)}{\sqrt[+]{\Delta_{{\l}}^2 (S(b, a)) - 4}} d\l & = & (-1)^N \cdot (-1)^{N+(N-k)+1} \delta_{N-k,N-n} \nonumber\\
& = & (-1)^{N+k+1} \delta_{kn}, \label{psinstildeliden}
\end{eqnarray}
thus satisfying the required identity (\ref{psinlsigmababedingung}).

By definition, $\psi_n ({\l})$ is a polynomial in ${\l}$ of maximal degree $N - 2$ for any $1 \leq n \leq N - 1$. Hence the map ${\l} \mapsto (- 1)^N \psi_{N - n} (- {\l})$ has the same property for any $1 \leq n \leq N - 1$. Together with (\ref{psinstildeliden}), this proves the lemma.
\end{proof}

In addition to the differentials $(\eta_n) _ {1 \leq n \leq N-1}$ defined by (\ref{psi}) and (\ref{omegalformel}), we need the special differentials
\begin{eqnarray}
 \Omega_1 & = & -\frac{1}{N} \frac{\dot{\Delta}_\l}{\sqrt{\Delta_\l^2 - 4}} d\l, \label{omega1def}\\
 \Omega_2 & = & -\frac{1}{N} \left( \frac{\l \dot{\Delta}_\l}{\sqrt{\Delta_\l^2 - 4}} d\l - \sum_{n=1}^{N-1} I_n \frac{\psi_n(\l)}{\sqrt{\Delta_\l^2 - 4}} \right). \label{omega2def}
\end{eqnarray}
We proved in \cite{ahtk6} the following lemma on $\Omega_1$ and $\Omega_2$:

\begin{lemma} \label{abeldiffexist}
The Abelian differentials $\Omega_1$ and $\Omega_2$ on $\Sigma_{b,a}$  have the following properties:
\begin{itemize}
	\item[(i)] $\Omega_1$ and $\Omega_2$ are holomorphic on $\Sigma_{b,a}$ except at $\infty^+$ and $\infty^-$ where in the standard charts, $\Omega_i$ admit an expansion of the form
\begin{equation} \label{omega1entwicklung}
 \Omega_1 = \mp \left( -\frac{1}{z} + e_1 + O(z) \right) dz \left( = \mp \left( \frac{1}{\l} - \frac{e_1}{\l^2} + O \left( \frac{1}{\l^3} \right) \right) d\l \right)
\end{equation} 
and
	\[ \Omega_2 = \mp \left( -\frac{1}{z^2} + f_1 + O(z) \right) dz \left( = \mp \left( 1 + O \left( \frac{1}{\l^2} \right) \right) d\l \right).
\]
\item[(ii)] $\Omega_1$ and $\Omega_2$ satisfy the normalization condtions
\begin{equation} \label{omegainorm}
	\int_{c_k} \Omega_i = 0 \quad \forall \, \1N1, \; i=1,2.
\end{equation}
\item[(iii)] When expressed in the local coordinate $\l$, on each of the two sheets, $\int_{E_{2N}}^\l \Omega_1$ admits an asymptotic expansion as $\l \to \infty$ ($\l$ real) of the form
\begin{equation} \label{intomega1}
	\int_{E_{2N}}^\l \Omega_1 = \mp \left( \log \l + e_0 + e_1 \frac{1}{\l} + \ldots \right).
\end{equation}
\end{itemize}
On $\Sigma_{b,a} \setminus \{ \l_1, \ldots, \l_{2N} \}$, $\Omega_1$ and $\Omega_2$ take the form $\Omega_i = \frac{\chi_i(\l)}{\sqrt{R(\l)}} \, d\l$ ($i=1,2$), where $\chi_i(\l)$ are polynomials in $\l$ of the form $\chi_1(\l) = \l^{N-1} + e \l^{N-2} + \ldots$ and 
$\chi_2(\l) = \l^{N} + f \l^{N-1} + \ldots$. Note that $\Omega_1$ and $\Omega_2$ do not depend on $\a$. Conversely, (i) and (ii) uniquely determine $\Omega_1$ and $\Omega_2$.
\end{lemma}

For any $1 \leq k \leq N-1$, introduce the $d_k$-periods of $\Omega_1$ and $\Omega_2$,
\begin{equation} \label{ukvkdef}
	U_k := \int_{d_k} \Omega_1; \qquad V_k := \int_{d_k} \Omega_2.
\end{equation}

In \cite{ahtk6}, we proved the following results on $U_k$ and $V_k$ for $\1N1$; recall from Theorem \ref{pertodabirkhoffthm} that $H_\alpha$ is the Toda Hamiltonian expressed in Birkhoff coordinates, $H \circ \Phi^{-1} = \frac{N\b^2}{2} + H_{\a}(I)$:
\begin{prop} \label{intdkomegai}
 For any $(b,a) \in \M^\bullet$ and any $\1N1$,
\begin{displaymath} 
  U_k = \frac{2 \pi i k}{N}, \quad V_k = \frac{2}{i} \omega_k.
\end{displaymath}
where $\omega_k$ is the Toda frequency $\omega_k = \partial_{I_k} H_\a$.
\end{prop}



For an arbitrary differential $\Omega$ on $\Sigma_{b,a}$ with poles at the points $P_j$ and corresponding residues $C_j$ (for $1 \leq j \leq m$), the $d_k$-periods are given by
\begin{equation} \label{intdlomega3}
 \int_{d_k} \Omega = 2 \pi i \sum_{j=1}^m C_j \int_{P_0}^{P_j} \eta_k
\end{equation}
for some base point $P_0$. This follows from the Riemann bilinear relations, for details see (\cite{toda}, section 4.4).

The next lemma establishes an important relation between the Dirichlet eigenvalues $(\mu_k)_{\1N1}$ and their ``shifted'' analogues $(\mu_k(j))_{\1N1}$ (recall the notation $\mu_k(j) = \mu_k(T^j(b,a))$). We will use Lemma \ref{munminus2muklemma} for the computation of $\theta_n(T(b,a))$ and $\theta_n(S(b,a))$ (see section \ref{symmactang}).

\begin{lemma} \label{munminus2muklemma}
 Let $(b,a) \in \M^\bullet$. For any $1 \leq n,j \leq N-1$,
\begin{equation} \label{munminus2muklemmaformel}
 \sum_{k=1}^{N-1} \int_{\mu_k^*}^{\mu_k^*(j)} \frac{\psi_{N-n}(\l)}{\sqrt{\Delta_{\l}^2-4}} \, d\l \;\; \textrm{mod} \;\; 2\pi = \frac{2 \pi j n}{N} + i \, \R.
\end{equation} 
\end{lemma}

\begin{proof}
This follows from a certain relation between $\mu_j(k)$ and $\mu_j(0) = \mu_j$ for any fixed indices $1 \leq j,l \leq N-1$,  namely
\begin{equation} \label{todamujkformel}
 \sum_{k=1}^{N-1} \int_{\mu_0}^{\mu_k(j)} \!\! \eta_l = j \int_{\infty^+}^{\infty^-} \!\! \eta_l + \sum_{k=1}^{N-1} \int_{\mu_0}^{\mu_k(0)} \!\! \eta_l - \sum_{k=1}^{N-1} n_k \tau_{lk} + m_l.
\end{equation}
Here $\mu_0$ is an arbitrarily chosen fixed point on the surface $\Sigma_{b,a}$, the coefficients $n_k$ and $m_l$ are integers, and the differentials $\eta_k$ and their $d_l$-periods $\tau_{lk}$ have been defined in (\ref{omegalformel}) and (\ref{taujldef}), respectively. The points $\mu_k(j)$ and $\mu_k(0)$ have to be taken on the same sheet of $\Sigma_{b,a}$ in order for (\ref{todamujkformel}) to hold.

The proof of (\ref{todamujkformel}) follows from an application of (\ref{intdlomega3}) to the differential
\begin{displaymath}
 \omega(k) = \left( \frac{d}{d\l} \ln y(k+1,\l) \right) d\l,
\end{displaymath}
where $y(\cdot,\l)$ is constructed from ``Bloch eigenfunctions'' of the difference equation (\ref{diff}), see \cite{data} or (\cite{toda}, section 4.3) for further details.



Setting $l=N-n$ in (\ref{todamujkformel}), using the definition (\ref{omegalformel}) of the differentials $\eta_l$, and eliminating the free base point $\mu_0$, we obtain
\begin{displaymath} 
 \sum_{k=1}^{N-1} \int_{\mu_k}^{\mu_k(j)} \frac{\psi_{N - n}(\l)}{\sqrt{\Delta_\l^2 - 4}} \, d\l = 2\pi \left( j \int_{\infty^+}^{\infty^-} \!\!\!\! \eta_{N-n} - \sum_{k=1}^{N-1} n_k \tau_{N-n,k} + m_{N-n} \right).
\end{displaymath}
Replacing $\mu_k$ and $\mu_k(j)$ by $\mu_k^*$ and $\mu_k^*(j)$, i.e. choosing the sheet of $\Sigma_{b,a}$ on which the square root is given by the starred root defined in (\ref{munstarred}), we obtain
\begin{equation} \label{todamujkformel2}
 \sum_{k=1}^{N-1} \int_{\mu_k^*}^{\mu_k^*(j)} \frac{\psi_{N - n}(\l)}{\sqrt{\Delta_\l^2 - 4}} \, d\l = 2\pi \left( j \int_{\infty^+}^{\infty^-} \!\!\!\! \eta_{N-n} - \sum_{k=1}^{N-1} n_k \tau_{N-n,k} + m_{N-n} \right).
\end{equation}
To prove Lemma \ref{munminus2muklemma}, it thus remains to compute the right hand side of (\ref{todamujkformel2}) mod $2 \pi$. We first compute $\int_{\infty^+}^{\infty^-} \eta_{N-n}$. 
It follows from Proposition \ref{intdkomegai} that
\begin{equation} \label{intdlomega1formel}
 \int_{d_l} \Omega_1 = \frac{2 \pi i \cdot l}{N}.
\end{equation}
On the other hand, by formula (\ref{omega1entwicklung}) from Lemma \ref{abeldiffexist}, $\Omega_1$ has poles at $\infty^+$ and $\infty^-$ with residues $1$ and $-1$, respectively. Therefore, by (\ref{intdlomega3}),
\begin{equation} \label{intdlomega1formel2}
 \int_{d_l} \Omega_1 = -2 \pi i \int_{\infty^+}^{\infty^-} \eta_l.
\end{equation} 
Hence for any $1 \leq j \leq N-1$, by setting $l=N-n$ in (\ref{intdlomega1formel}) and (\ref{intdlomega1formel2}) we obtain
\begin{equation} \label{intdlomega1formel3}
 j \int_{\infty^+}^{\infty^-} \eta_{N-n} = -\frac{1}{2 \pi i} \cdot j \cdot \int_{d_{N-n}} \Omega_1 = -\frac{1}{2 \pi i} \cdot j \cdot \frac{2 \pi i (N-n)}{N} = -\frac{j (N-n)}{N}.
\end{equation}
It follows that
\[ 2 \pi j \int_{\infty^+}^{\infty^-} \eta_{N-n} \; \textrm{mod} \; 2 \pi = -\frac{2 \pi j (N-n)}{N} \; \textrm{mod} \; 2 \pi = \frac{2 \pi j n}{N} \; \textrm{mod} \; 2 \pi. \]
Returning to the right hand side of (\ref{todamujkformel2}), recall from (\ref{todamujkformel}) that $m_{N-n}$ is an integer, hence $2\pi \, m_{N-n} \equiv 0$ mod $2\pi$. Moreover, recall from (\ref{taujkiniR}) that $\t_{lk} \in i \R$ for any $1 \leq l,k \leq N-1$. 
Hence, since the $n_k$'s are real, it follows that
\begin{equation} \label{intdlomega1formel4}
 2\pi \left( -\sum_{k=1}^{N-1} n_k \tau_{N-n,k} + m_{N-n} \right) \; \textrm{mod} \; 2\pi \in i \R.
\end{equation} 
Combining (\ref{intdlomega1formel3}) and (\ref{intdlomega1formel4}), we obtain (\ref{munminus2muklemmaformel}) and therefore Lemma \ref{munminus2muklemma}.
\end{proof}

\section{Symmetries in action-angle variables} \label{symmactang}

\paragraph{Actions}

Recall from \cite{ahtk1} (see also \cite{ahtk2}) that, for $(b,a) \in \M$, the action variables $(I_k)_{1 \leq k \leq N-1}$ are given by
\begin{equation} \label{ikbaformel}
 I_k(b,a) = \frac{1}{\pi} \int_{\l_{2k}(b,a)}^{\l_{2k+1}(b,a)} \textrm{arcosh } \left| \frac{\Delta_\l(b,a)}{2} \right| \, d\l.
\end{equation}

\begin{prop} \label{iksbaprop}
 For any $(b,a) \in \M$ and any $\1N1$,
\begin{eqnarray}
 (i) \; I_k(T(b,a)) & = & I_{k}(b,a),
\label{iktjba} \\
 (ii) \; I_k(S(b,a)) & = & I_{N-k}(b,a).
\label{iktildesba}
\end{eqnarray} 
\end{prop}

\begin{proof}
Formula (\ref{iktjba}) follows from the fact that all quantities appearing on the right hand side of the definition (\ref{ikbaformel}) of $I_k(b,a)$ (i.e. $\l_{2k}$, $\l_{2k+1}$, and $\Delta_\l$) are invariant under the transformation $T$ (see Lemma \ref{deltalambdatranslemma}). We thus turn to (\ref{iktildesba}).

Applying (\ref{ikbaformel}) to $S(b,a)$, it follows that
\begin{displaymath}
 I_k(S(b,a)) = \frac{1}{\pi} \int_{\l_{2k}(S(b,a))}^{\l_{2k+1}(S(b,a))} \textrm{arcosh } \left| \frac{\Delta_\l(S(b,a))}{2} \right| \, d\l.
\end{displaymath}
By Lemma \ref{ljstildelemma} and Corollary \ref{deltalstildecor}, it follows that
\begin{displaymath}
 I_k(S(b,a)) = \frac{1}{\pi} \int_{-\l_{2N-2k+1}(b,a)}^{-\l_{2N-2k}(b,a)} \textrm{arcosh } \left| \frac{\Delta_{-\l}(b,a)}{2} \right| \, d\l.
\end{displaymath}
Making in the integral the substitution $\l \mapsto -\l$ and reversing the integration direction, we obtain
\begin{eqnarray*}
  I_k(S(b,a)) & = & \frac{1}{\pi} \int_{\l_{2N-2k}(b,a)}^{\l_{2N-2k+1}(b,a)} \textrm{arcosh } \left| \frac{\Delta_\l(b,a)}{2} \right| \, d\l \\
& = & I_{N-k}(b,a),
\end{eqnarray*}
again by the definition (\ref{ikbaformel}) of the action variables. This proves (\ref{iktildesba}).
\end{proof}

\paragraph{Angles}

To prove (\ref{thetaktildes}), we proceed in a similar fashion. We recall from \cite{ahtk1} that the angle variables $(\theta_n)_{1 \leq n \leq N-1}$ are given by $\theta_n = \sum_{k=1}^N \beta_k^n$ mod $2\pi$, where
\begin{equation} \label{betaknallgformel}
 \beta_k^n = \int_{\l_{2k}}^{\mu_k^*} \frac{\psi_n(\l)}{\sqrt{\Delta_\l^2-4}} d\l \; (k \neq n), \quad \beta_n^n = \int_{\l_{2n}}^{\mu_n^*} \frac{\psi_n(\l)}{\sqrt{\Delta_\l^2-4}} \, d\l \; \textrm{mod} \; 2\pi.
\end{equation}
Note that these integrals are integrals on the Riemann surface $\Sigma_{b,a}$, and $\theta_n(b,a)$ is only defined in the case $I_n(b,a) \neq 0$.

\begin{prop} \label{thksbaprop}
 For any $(b,a) \in \M^\bullet$ and any $1 \leq n \leq N-1$, the angles $\theta_n(T(b,a))$ and $\theta_n(S(b,a))$ are given by
\begin{eqnarray}
 (i) \; \theta_n(T(b,a)) & = & \theta_n(b,a) + \frac{2 \pi n}{N} \quad \textrm{mod} \; 2\pi, \label{thetantba}\\
 (ii) \; \theta_n(S(b,a)) & = & \theta_{N-n}(b,a) + \pi - \frac{4 \pi n}{N} \quad \textrm{mod} \; 2\pi. \label{thetansba}
\end{eqnarray} 
\end{prop}

\begin{proof}
 By (\ref{betaknallgformel}),
\begin{equation} \label{betakntba}
 \beta_k^n(T(b,a)) = \int_{\l_{2k}(T(b,a))}^{\mu_k^*(T(b,a))} \frac{\psi_n(\l)(T(b,a))}{\sqrt{\Delta_\l^2(T(b,a))-4}} \, d\l.
\end{equation}
and
\begin{equation} \label{betaknsba}
 \beta_k^n(S(b,a)) = \int_{\l_{2k}(S(b,a))}^{\mu_k^*(S(b,a))} \frac{\psi_n(\l)(S(b,a))}{\sqrt{\Delta_\l^2(S(b,a))-4}} \, d\l.
\end{equation}

Let us first compute $\beta_k^n(T(b,a))$. By Lemma \ref{deltalambdatranslemma} and Lemma \ref{psinlambdatildetslemma} (i), we can rewrite (\ref{betakntba}) as
\begin{displaymath}
 \beta_k^n(T(b,a)) = \int_{\l_{2k}(b,a)}^{\mu_k^*(T(b,a))} \frac{\psi_n(\l)(b,a)}{\sqrt{\Delta_\l^2(b,a)-4}} \, d\l.
\end{displaymath}
Writing $\mu_k(1) = \mu_k(T^1(b,a)) = \mu_k(T(b,a))$, it follows that
\begin{displaymath}
 \beta_k^n(T(b,a)) = \left( \int_{\l_{2k}}^{\mu_k^*} + \int_{\mu_k^*}^{\mu_k^*(1)} \right) \frac{\psi_n(\l)(b,a)}{\sqrt{\Delta_\l^2(b,a)-4}} \, d\l
\end{displaymath}
Therefore
\begin{displaymath}
 \beta_k^n(T(b,a)) = \beta_k^n(b,a) + \int_{\mu_k^*}^{\mu_k^*(1)} \frac{\psi_n(\l)(b,a)}{\sqrt{\Delta_\l^2(b,a)-4}} \, d\l.
\end{displaymath}
Taking the sum over all indices $\1N1$, we obtain, using the definition of the angle variable $\theta_n$,
\begin{equation} \label{thetantba2}
 \theta_n(T(b,a)) = \theta_n(b,a) + \sum_{k=1}^{N-1} \int_{\mu_k^*}^{\mu_k^*(1)} \frac{\psi_n(\l)(b,a)}{\sqrt{\Delta_\l^2(b,a)-4}} \, d\l \quad \textrm{mod} \; 2\pi.
\end{equation}
We see that the sum on the right hand side of (\ref{thetantba2}) is real, hence by Lemma \ref{munminus2muklemma} (in the case $j=1$), this sum (mod $2\pi$) equals $\frac{2 \pi n}{N}$, and therefore
\begin{displaymath}
 \theta_n(T(b,a)) = \theta_n(b,a) + \frac{2 \pi n}{N} \quad \textrm{mod} \; 2\pi,
\end{displaymath}
as claimed. This proves (\ref{thetantba}).

We now turn to the proof of (\ref{thetansba}). Reformulating (\ref{betaknsba}) as a real integral, we have
\begin{displaymath}
 \beta_k^n(S(b,a)) = \e_{S}(k) \int_{\l_{2k}(S(b,a))}^{\mu_k(S(b,a))} \frac{\psi_n(\l)(S(b,a))}{\sqrt[+]{\Delta_\l^2(S(b,a))-4}} \, d\l,
\end{displaymath}
where $\e_{S}(k) \in \{ \pm 1 \}$, and where we denote by $\sqrt[+]{\cdot}$ the standard square root on $\R_{\geq 0}$. The sign of $\e_{S}(k)$ is the sign of the starred square root $\sqrt[*]{\Delta^2_{\mu_n}(S(b,a))-4}$ and will be discussed below. Further we compute
\begin{eqnarray*}
  (-1)^N \beta_k^n(S(b,a)) & = & \e_{S}(k) \int_{\l_{2k}(S(b,a))}^{\mu_k(S(b,a))} \frac{\psi_{N-n}(-\l)(b,a)}{\sqrt[+]{\Delta_{-\l}^2(b,a)-4}} \; d\l \\
& = & \e_{S}(k) \int_{-\l_{2N+1-2k}(b,a)}^{-\mu_{N-k}\left(T^{N-2}(b,a)\right)} \frac{\psi_{N-n}(-\l)(b,a)}{\sqrt[+]{\Delta_{-\l}^2(b,a)-4}} \; d\l \\
& = & \e_{S}(k) \int_{\mu_{N-k}(T^{N-2}(b,a))}^{\l_{2N+1-2k}(b,a)} \frac{\psi_{N-n}(\l)(b,a)}{\sqrt[+]{\Delta_{\l}^2(b,a)-4}} \; d\l,
\end{eqnarray*}
using Corollary \ref{deltalstildecor} and Lemma \ref{psifnstildecor} in the first, Lemmas \ref{ljstildelemma} and \ref{munstildelemma} in the second, and a substitution $\l \mapsto -\l$ together with a reverse integration order in the third step.

Next we replace $\l_{2N+1-2k} (b,a)$ by $\l_{2N+1-2k}\left(T^{N-2}(b,a)\right)$, and analogously for $\psi_{N-n}(\l)(b,a)$ and $\Delta_{\l}^2(b,a)-4$. By Lemmas \ref{deltalambdatranslemma} and \ref{psinlambdatildetslemma}, this leaves all these quantities invariant. Hence
\begin{displaymath}
 (-1)^N \beta_k^n(S(b,a)) = \e_{S}(k) \int_{\mu_{N-k}\left(T^{N-2}(b,a)\right)}^{\l_{2N+1-2k}\left(T^{N-2}(b,a)\right)} \frac{\psi_{N-n}(\l)\left(T^{N-2}(b,a)\right)}{\sqrt[+]{\Delta_{\l}^2\left(T^{N-2}(b,a)\right)-4}} \; d\l.
\end{displaymath}
It follows that we have, writing $(N-2)$ instead of $T^{N-2}(b,a)$ and using (\ref{psi2}),
\setlength\arraycolsep{1pt}{\begin{eqnarray*}
 (-1)^N \beta_k^n(S(b,a)) & = & \e_{S}(k) \left( \int_{\l_{2N-2k}(N-2)}^{\l_{2N+1-2k}(N\!-\!2)} \!\!\!-\! \int_{\l_{2N-2k}(N-2)}^{\mu_{N-k}(N-2)} \right) \! \frac{\psi_{N-n}(\l)(N\!-\!2)}{\sqrt[+]{\Delta_{\l}^2(N \!-\! 2) \!-\! 4}} \; d\l \\
& = & \pm \pi \delta_{N-k,N-n} - \e_{S}(k) \int_{\l_{2N-2k}(N-2)}^{\mu_{N-k}(N-2)} \frac{\psi_{N-n}(\l)(N-2)}{\sqrt[+]{\Delta_{\l}^2(N-2)-4}} \; d\l,
\end{eqnarray*}}
Up to an integer multiple of $2\pi$ (and since $\delta_{N-k,N-n} = \delta_{kn}$), it follows that
\begin{eqnarray*}
 (-1)^N \beta_k^n(S(b,a)) & = & \pi \delta_{kn} - \e_{S}(k) \int_{\l_{2N-2k}(N-2)}^{\mu_{N-k}(N-2)} \frac{\psi_{N-n}(\l)(N-2)}{\sqrt[+]{\Delta_{\l}^2(N-2)-4}} \; d\l \\
& = & \pi \delta_{kn} -\e_{S}(k) \e_{T^{N\!-\!2}}(N\!-\!k) \! \int_{\l_{2N-2k}(N\!-\!2)}^{\mu_{N-k}^*(N\!-\!2)} \! \frac{\psi_{N-n}(\l)(N\!-\!2)}{\sqrt{\Delta_{\l}^2(N\!-\!2)\!-\!4}} \; d\l,
\end{eqnarray*}
where the value of $\e_{T^{N-2}}(N-k) \in \{ \pm 1 \}$ depends on the sign of the starred square root $\sqrt[*]{\Delta_{\mu_{N-k}}^2(T^{N-2}(b,a))-4}$. Note that the last integral is again an integral on the Riemann surface $\Sigma_{b,a}$ (recall from Lemma \ref{psifnstildecor} that $\Sigma_{T(b, a)} = \Sigma_{b, a}$).

In order to gain information on $\e_{S}(k) \e_{T^{N-2}}(N-k)$, we conclude from (\ref{vorzeichen3formel}) that 
\begin{displaymath}
 \e_{\tilde{S}}(k) = (-1)^{N+1} \e_{T^{N-2}}(N-k).
\end{displaymath}
Together with the fact that $\e_{\tilde{S}}(k), \e_{T^{N-2}}(N-k) \in \{ \pm 1 \}$, this shows that
\begin{displaymath} 
  \e_{S}(k) \e_{T^{N-2}}(N-k) = (-1)^{N+1}.
 \end{displaymath}
It follows that 
\begin{eqnarray*}
 \beta_k^n(S(b,a)) & = & \pi \delta_{kn} + \int_{\l_{2N-2k}(N-2)}^{\mu_{N-k}^*(N-2)} \frac{\psi_{N-n}(\l)(N-2)}{\sqrt{\Delta_{\l}^2(N-2)-4}} \; d\l \\
& = & \pi \delta_{kn} + \int_{\l_{2N-2k}}^{\mu_{N-k}^*(N-2)} \frac{\psi_{N-n}(\l)}{\sqrt{\Delta_{\l}^2-4}} \; d\l,
\end{eqnarray*}
again using the invariance of the quantities $\Delta_\l$, $\l_{2N-2k}$, and $\psi_j(\l)$ under the shift map $T$ (Lemma \ref{deltalambdatranslemma}, Lemma \ref{psinlambdatildetslemma}). Next,
\begin{eqnarray*}
 \beta_k^n(S(b,a)) & = & \pi \delta_{kn} + \int_{\l_{2N-2k}}^{\mu_{N-k}^*(N-2)} \frac{\psi_{N-n}(\l)}{\sqrt{\Delta_{\l}^2-4}} \; d\l \\
& = & \pi \delta_{kn} + \left( \int_{\l_{2N-2k}}^{\mu_{N-k}^*} + \int_{\mu_{N-k}^*}^{\mu_{N-k}^*(N-2)} \right) \frac{\psi_{N-n}(\l)}{\sqrt{\Delta_{\l}^2-4}} \; d\l \\
& = & \beta_{N-k}^{N-n} (b,a) + \pi \delta_{kn} + \int_{\mu_{N-k}^*}^{\mu_{N-k}^*(N-2)} \frac{\psi_{N-n}(\l)}{\sqrt{\Delta_{\l}^2-4}} \; d\l. 
\end{eqnarray*}
By definition of the angle variables $\theta_n$, $\theta_n(S(b,a)) = \sum_{k=1}^{N-1} \beta_k^n(S(b,a))$ mod $2\pi$. Therefore, mod $2\pi$,
\begin{eqnarray}
 \theta_n(S(b,a)) & = & \sum_{k=1}^{N-1} \left( \pi \delta_{kn} + \beta_{N-k}^{N-n} (b,a) + \int_{\mu_{N-k}^*}^{\mu_{N-k}^*(N-2)} \frac{\psi_{N-n}(\l)}{\sqrt{\Delta_{\l}^2-4}} \; d\l \right)  \nonumber\\
& = & \pi + \sum_{k=1}^{N-1} \beta_{N-k}^{N-n} (b,a) + \sum_{k=1}^{N-1} \left( \int_{\mu_{N-k}^*}^{\mu_{N-k}^*(N-2)} \frac{\psi_{N-n}(\l)}{\sqrt{\Delta_{\l}^2-4}} \; d\l  \nonumber\right) \\
& = & \theta_{N-n}(b,a) + \pi + \sum_{k=1}^{N-1} \left( \int_{\mu_k^*}^{\mu_k^*(N-2)} \frac{\psi_{N-n}(\l)}{\sqrt{\Delta_{\l}^2-4}} \; d\l \right). \label{betaknbetanknn2}
\end{eqnarray}
To evaulate the sum on the right hand side of (\ref{betaknbetanknn2}), we first conclude from (\ref{betaknbetanknn2}) that this sum is real and then use Lemma \ref{munminus2muklemma} (for $j=N-2$) to evaulate it. We obtain, again mod $2\pi$,
\begin{displaymath}
 \theta_n(S(b,a)) = \theta_{N-n}(b,a) + \pi + \frac{2 \pi (N-2) n}{N} = \theta_{N-n}(b,a) + \pi - \frac{4 \pi n}{N}.
\end{displaymath}
This proves (\ref{thetansba}) and completes the proof of Proposition \ref{thksbaprop}.
\end{proof}

\begin{proof}[Proof of Theorem \ref{ikthetaktildesba}]
The statements (\ref{iktilder}) and (\ref{iktildes}) on the action variables of $T(b,a)$ and $S(b,a)$ follow from Proposition \ref{iksbaprop}. In the case $(b,a) \in \M^\bullet$, and therefore also $S(b,a) \in \M^\bullet$, the statements (\ref{thetaktilder}) and (\ref{thetaktildes}) on the angle variables of $T(b,a)$ and $S(b,a)$ follow from Proposition \ref{thksbaprop}. In the general case, the same holds by by a continuity argument, since $\M^\bullet$ is dense in $\M \setminus D_n$, and the angle variables are continuous functions on their domains of definition.
\end{proof}

\begin{proof}[Proof of Theorem \ref{xkxnkproof}]
To prove (\ref{szomega01omega01stilde}), we have to show that for any $1 \leq |k| \leq N-1$, we have the identities $\left( \sz(\z(b,a)) \right)_k = \z_k\left( S(b,a) \right)$. 
Using the definition (\ref{szdef}) of the map $\sz$, this amounts to proving that for any $1 \leq |k| \leq N-1$,
\begin{equation} \label{zetankzktildes}
 -e^{4 \pi i k / N} \z_{N-k}(b,a) = \z_k\left( S(b,a) \right).
\end{equation}
From the formulas (\ref{iktildes}) and (\ref{thetaktildes}) from Theorem \ref{ikthetaktildesba} for the actions $I_k\left( S(b,a) \right)$ and the angles $\theta_k\left( S(b,a) \right)$, respectively, and using
\begin{equation} \label{xkykikthetak}
\left\{ \begin{array}{rcl}
\z_k & = & \sqrt{2 I_k} e^{-i \theta_k} \\
\z_{-k} & = & \sqrt{2 I_k} e^{i \theta_k}
\end{array} \right.,
\end{equation}
(a consequence of (\ref{xnynsqrtinthetan}) and (\ref{complexdef})), it follows that
\begin{eqnarray*}
\zeta_k \left( S(b,a) \right) & = & \sqrt{2 I_k\left( S(b,a) \right)} e^{-i \theta_k\left( S(b,a) \right)}\\
& = & \sqrt{2 I_{N-k}(b,a)} \cdot e^{-i \left( \theta_{N-k}(b,a) + \pi - 4 \pi k / N \right)} \\
& = & -e^{4 \pi i k / N} \cdot \sqrt{2 I_{N-k}(b,a)} \cdot e ^{-i \theta_{N-k}(b,a)} \\
& = & -e^{4 \pi i k / N} \cdot \zeta_{N-k}(b,a).
\end{eqnarray*}
This is the required identity (\ref{zetankzktildes}).
%
\end{proof}

\begin{cor} \label{fixsspectrumlemma}
 For any $(b,a) \in$ Fix$(S)$,
\begin{itemize}
\item[(i)] $I_k(b,a) = I_{N-k}(b,a)$ for any $\1N1$.
\item[(ii)] $\l_j(b,a) = -\l_{2N+1-j}(b,a)$ for any $1 \leq j \leq 2N$.
\item[(iii)] $I_{\frac{N}{2}}(b,a) = 0$ and $\l_N(b,a) = \l_{N+1}(b,a) = 0$, if $N$ is even.
\end{itemize}
\end{cor}

\begin{proof}
The statement (i) directly follows from the identity (\ref{iktildesba}) of Proposition \ref{iksbaprop}, and (ii) is an immediate consequence of the identity (\ref{ljtildesba1}) of Lemma \ref{ljstildelemma}. 
To prove (iii), note that it follows from Theorem \ref{xkxnkproof} that by the Birkhoff map $\Phi$, the set Fix$(S)$ is mapped to the set Fix$(\sz) \subseteq \mathcal{Z}$, described (by the definition (\ref{szdef}) of $\sz$) by the equations
\begin{equation} \label{fixscondzeta}
 \z_k = -e^{-4 \pi i k / N} \z_{N-k} \quad \forall 1 \leq |k| \leq N-1.
\end{equation}
In the case $k = \frac{N}{2}$, (\ref{fixscondzeta}) becomes $\z_{N/2} = -\z_{N/2}$ and hence $\z_{N/2} = 0$. By the definition (\ref{complexdef}) of the $\z_k$'s, the action variables $(I_k)_{1 \leq k \leq N-1}$ can be written as $I_k = \z_k \z_{-k} = |\z_k|^2$, and thus $\z_{N/2} = 0$ implies $I_{N/2} = 0$. From this it follows directly that $\l_{N} = \l_{N+1}$ (recall that $I_n = 0$ iff $\g_n = 0$). By (ii), we know that $\l_N = -\l_{N+1}$, and $\l_{N} = \l_{N+1} = 0$ follows.
\end{proof}

A sharper version of Corollary \ref{fixsspectrumlemma} is the following characterization of the set Fix$(S)$ in terms of the action and angle variables $I_k$ and $\theta_k$.

\begin{prop} \label{fixsikthetakprop}
Let $\Phi$ denote the Birkhoff map of Theorem \ref{pertodabirkhoffthm}, and let Fix$(S)_\a = \textrm{Fix}(S) \cap \M_{0,\a}$ for any fixed $\a > 0$. Then
 \begin{eqnarray}
  \textrm{Fix}(S)_\a & = & \left\{ \Phi^{-1} \left( \left( \sqrt[+]{2 I_k} e^{i \theta_k} \right)_{\1N1}, 0, \a \right) \right| I_{N-k} = I_k \; \forall \1N1,  \nonumber\\ && \qquad \qquad \left. \theta_{N-k} = \theta_k + \pi + \frac{4 \pi k}{N} \; \textrm{mod} \; 2\pi \quad \forall \,k: I_k \neq 0 \right\}. \label{fixsmenge}
 \end{eqnarray}
\end{prop}

\begin{proof}
 Let $P$ denote the set on the right hand side of (\ref{fixsmenge}). That Fix$(S) \subseteq P$ follows in the same way as Lemma \ref{fixsspectrumlemma} (i) from the identites (\ref{iktildesba}) and (\ref{thetansba}) of Propositions \ref{iksbaprop} and \ref{thksbaprop}. Conversely, let $(b,a) \in P$. Then, by Theorem \ref{ikthetaktildesba}, $I_k(S(b,a)) = I_k(b,a)$ for any $\1N1$ and $\theta_k(S(b,a)) = \theta_k(b,a)$ for any $\1N1$ with $I_k(b,a) \neq 0$. Moreover $C_i(b,a) = C_i(S(b,a))$ for $i=1,2$. By the bijectivity of $\Phi$, it follows that $(b,a) = S(b,a)$.
\end{proof}



\section[Dirichlet Toda lattice]{Global Birkhoff Coordinates for the Dirichlet Toda lattice} \label{dirbirkhoff}

As already mentioned in the introduction, besides the periodic lattice with $N$ particles we also consider the lattice with Dirichlet boundary conditions, i.e. the lattice with $N'$ moving particles ($N' \geq 2$), the Hamiltonian given by (\ref{hamdirtodapq}) and the boundary conditions (\ref{dirbdycond}).
We treat the lattice with boundary conditions (\ref{dirbdycond}) similarly to the periodic lattice, with the goal of embedding its phase space into the phase space of the periodic lattice (this will be accomplished by the map $\Theta^{(D)}$ defined below in (\ref{ThetaDdef})). As a first step, following Flaschka \cite{fla1} and as in the periodic case, we introduce the (noncanonical) coordinates
\begin{equation} \label{flaschkadirdef}
d_n := -p_n \in \R \quad (0 \! \leq \! n \! \leq \! N'), \quad c_n := \g e^{\frac{1}{2} (q_n - q_{n+1})} \in \R_{>0} \quad (0 \! \leq \! n \! \leq \! N').
\end{equation}
The boundary conditions (\ref{dirbdycond}) imply that $d_0=0$ and $\prod_{n=0}^{\D} c_n = \g^{\D+1}$; we will identify related Casimir functions below.

In these coordinates the Hamiltonian $H_{Toda}^{(D)}$ takes the simple form
\begin{equation} \label{htodadir}
H^{(D)} = \frac{1}{2} \sum_{n=1}^{N'} d_n^2 + \sum_{n=0}^{N'} c_n^2,
\end{equation}
and the equations of motion are
\begin{equation} \label{flaeqndir}
\left\{ \begin{array}{lllll}
 \dot{d}_0 & = & 0, \\
 \dot{d}_n & = & c_n^2 - c_{n-1}^2 \qquad (1 \leq n \leq \D), \\
 \dot{c}_0 & = & \frac{1}{2} \, c_0 d_1, \\
 \dot{c}_n & = & \frac{1}{2} \, c_n (d_{n+1} - d_n) \qquad (1 \leq n \leq \D-1), \\
 \dot{c}_{\D} & = & -\frac{1}{2} \, c_{\D} d_{\D}.
\end{array} \right.
\end{equation}

Similarly to the periodic case, we study the system of equations (\ref{flaeqndir}) on the $2(\D+1)$-dimensional phase space
\begin{displaymath}
\M^{(D)} := \R^{\D+1} \times \R_{>0}^{\D+1}.
\end{displaymath}
This system is Hamiltonian with respect to the nonstandard Poisson structure $J^{(D)} \equiv J^{(D)}_{d,c}$, defined at a point $(d,c) = (d_n, c_n)_{0 \leq n \leq \D}$ by
\begin{equation} \label{jdirdef}
J^{(D)} = \left( \begin{array}{cc}
0 & A^{(D)} \\
-\left( A^{(D)} \right)^T & 0 \\
\end{array} \right),
\end{equation}
where $A^{(D)}$ is the $d$-independent $(\D+1) \times (\D+1)$-matrix
\begin{equation} \label{adirdef}
A^{(D)} = \frac{1}{2} \left( \begin{array}{ccccc}
0 & 0 & \ldots & 0 & 0 \\
-c_0 & c_1 & 0 & \ddots & 0 \\
0 & -c_1 & c_2 & \ddots & \vdots \\
\vdots & \ddots & \ddots & \ddots & 0 \\
0 & \ldots & 0 & -c_{\D-1} & c_{\D} \\
\end{array} \right).
\end{equation}
The Poisson bracket corresponding to (\ref{jdirdef}) is then given by
\begin{eqnarray}
\{ F, G \}_{J^{(D)}}(d,c) & = & \langle (\nabla_d F, \nabla_c F), \, J^{(D)} \, (\nabla_d G, \nabla_c G) \rangle_{\R^{2\D+2}} \nonumber\\
& = & \langle \nabla_d F, A \, \nabla_c G \rangle_{\R^{\D+1}} - \langle \nabla_c F, A^T \, \nabla_b G \rangle_{\R^{\D+1}}. 
\end{eqnarray}
where $F,G \in C^1(\M^{(D)})$ and where $\nabla_d$ and $\nabla_c$ denote the gradients with respect to the $(\D+1)$-vectors $d = (d_0, \ldots, d_{\D})$ and $c = (c_0, \ldots, c_{\D})$, respectively. Therefore, the equations (\ref{flaeqndir}) can alternatively be written as $\dot{d}_n = \left\{ d_n, H^{(D)} \right\}_{J^{(D)}}$, $\dot{c}_n = \left\{ c_n, H^{(D)} \right\}_{J^{(D)}}$ ($1 \leq n \leq \D$).

Since the matrix $A^{(D)}$ defined by (\ref{adirdef}) has rank $\D$, the Poisson structure $J^{(D)}$ is degenerate. It admits the two Casimir functions
\begin{equation} 
E_1 := d_0 \quad \textrm{and} \quad E_2 := \left( \prod_{n=0}^{\D} c_n \right)^\frac{1}{\D+1}
\end{equation}
whose gradients $\ndc E_i \! = \! \left( \n_d E_i, \n_c E_i \right)$ ($i = 1,2$), given by
\setlength\arraycolsep{1.5pt} {
\begin{eqnarray}
\nabla_d E_1 & = & (1, 0, \ldots, 0), \qquad \nabla_c E_1 = 0, \label{c1graddir} \\
\nabla_d E_2 & = & 0, \qquad \nabla_c E_2 = \frac{C_2}{\D+1} \left( \frac{1}{c_0}, \ldots, \frac{1}{c_{\D}} \right), \label{c2graddir}
\end{eqnarray}}
are linearly independent at each point $(d,c)$ of $\M^{(D)}$.

Let
\begin{displaymath}
\Mdc^{(D)} := \left\{ (d,c) \in \M^{(D)} : \left( E_1, E_2 \right) = (\delta, \g) \right\}
\end{displaymath}
denote the level set of $\left( E_1, E_2 \right)$ for $(\delta, \g) \in \R \times \R_{>0}$. 
By (\ref{c1graddir})-(\ref{c2graddir}), the sets $\Mdc^{(D)}$ are real analytic submanifolds of $\M^{(D)}$ of codimension two. Furthermore the Poisson structure $J^{(D)}$, restricted to $\Mdc^{(D)}$, becomes nondegenerate everywhere on $\Mdc^{(D)}$ and therefore induces a symplectic structure $\nu_{\delta, \g}^{(D)}$ on $\Mdc^{(D)}$. In this way, we obtain a symplectic foliation of $\M^{(D)}$ with $\Mdc^{(D)}$ being its (symplectic) leaves.
Note that we are mainly interested in the case $\delta = 0$, i.e. the set $\M_{0,\g}^{(D)}$, for the assumption $\delta \neq 0$ contradicts the boundary conditions (\ref{dirbdycond}). We have included the case of general $\delta$'s in order to have an even-dimensional phase space. In the sequel, we write $\Hdc^{(D)} = H^{(D)}|_{\Mdc^{(D)}}$.

To state the main result of this section, we introduce the model space
$$ \P^{(D)} := \R^{2N'} \times \R \times \R_{>0} $$
 endowed with the degenerate Poisson structure $J_0^{(D)}$ whose symplectic leaves are $\R^{2N'} \times \{ \d \} \times \{ \g \}$ endowed with the canonical symplectic structure.
\begin{theorem} \label{dirtodathm}
There exists a map
\begin{displaymath}
\begin{array}{ccll}
 \Phi^{(D)}: & \left( \M^{(D)}, J^{(D)} \right) & \to & \left( \P^{(D)}, J_0^{(D)} \right) \\
 & (d,c) & \mapsto & ((x_n, y_n)_{1 \leq n \leq N'}, E_1, E_2)
\end{array}
\end{displaymath}
with the following properties:
\begin{itemize}
  \item $\Phi^{(D)}$ is a real analytic diffeomorphism.
  \item $\Phi^{(D)}$ is canonical, i.e. it preserves the  Poisson brackets. In particular, the symplectic foliation of $\M^{(D)}$ by $\Mdc^{(D)}$ is trivial.
  \item The coordinates $(x_n, y_n)_{1 \leq n \leq N'}, E_1, E_2$ are
  global Birkhoff coordinates for the Toda lattice with Dirichlet boundary conditions, i.e. the transformed Toda Hamiltonian $\hat{H}^{(D)} = H^{(D)} \circ \left( \Omega^{(D)} \right)^{-1}$
  is a function of the actions $(x_n^2 + y_n^2)/2$ $(1 \leq n \leq N')$ and $E_1, E_2$ alone.
\end{itemize}
\end{theorem}

The crucial point for proving Theorem \ref{dirtodathm} is the observation that the Toda lattice with $N'$ moving particles and \emph{Dirichlet} boundary conditions can be regarded as an invariant submanifold of the lattice with \emph{periodic} boundary conditions and $N = 2N'+2$ particles, analogously to the case of general FPU chains, in which case this observation was made by Rink - see \cite{rink06}. In this way, we can use the analogous theorem for the periodic lattice, Theorem \ref{pertodabirkhoffthm}, to prove Theorem \ref{dirtodathm}.

More precisely, we embed the phase space $\M^{(D)}$ of the Dirichlet lattice with $N'$ moving particles into the phase space of the periodic lattice with $N=2N'+2$ particles by the map
\begin{equation} \label{ThetaDdef}
\begin{array}{clll}
 \Theta^{(D)}: & \M^{(D)} & \to & \M \\
 & (d,c) = \left( d_j, c_j \right)_{0 \leq j \leq \D} & \mapsto & \Theta(d,c) = \left( b_j, a_j \right)_{1 \leq j \leq N},
\end{array}
\end{equation}
where
\begin{equation} \label{ThetaDdefcomponents}
	 b_j = \frac{1}{\sqrt{2}} \left\{ \begin{array}{l} d_j \\ 0 \\ -d_{N-j} \\ 0 \end{array} \right., \quad a_j = \frac{1}{\sqrt{2}} \left\{ \begin{array}{ll} c_j & (1 \leq j \leq \D) \\ c_{\D} & (j=\D+1) \\ c_{N-j-1} & (\D+2 \leq j \leq 2\D+1) \\ c_0 & (j=2\D+2) \end{array} \right..
\end{equation}
The connection to the map $\tilde{S}$ introduced in (\ref{tildesdef}) becomes clear through the following observation; as before we omit the tilde and write $S$ instead of $\tilde{S}$.

\begin{lemma} \label{flaschkalemma}
 Let $(\d,\g) \in \R \times \R_{>0}$.
\begin{itemize}
 \item[(i)] 
$\Theta^{(D)}\left( \Mdc^{(D)} \right) = \M_{0,\frac{\g}{\sqrt{2}}} \cap \textrm{Fix}(S)$.
\item[(ii)] The Hamiltonians $H$ and $H^{(D)}$ of the Toda lattice in Flaschka variables with periodic and Dirichlet boundary conditions, respectively, given by (\ref{htodaperflaschka}) and (\ref{htodadir}), satisfy $H \circ \Theta^{(D)} = \HD$.
\end{itemize}
\end{lemma}

\begin{proof}
Both statements can be checked by direct computations, using the definition (\ref{ThetaDdefcomponents}) of the map $\Theta^{(D)}$.
\end{proof}

In other words, by (i), for any $(\d,\g) \in \R \times \R_{>0}$, $\Theta^{(D)}|_{\Mdc^{(D)}} =: \Thdc^{(D)}$ is a parametrization of $\M_{0,\frac{\g}{\sqrt{2}}} \cap \textrm{Fix}(S)$, and by (ii), we have the commutative diagram
\begin{displaymath}
 \xymatrixcolsep{3pc}\xymatrix{\Mba \ar@/^/[dr]^H &  \\ 
           \Mdc^{(D)} \ar[r]_{\HD} \ar[u]^{\Theta^{(D)}} & \R }
\end{displaymath}
Moreover, $\Thdc^{(D)}$ is a canonical map from $\left( \Mdc^{(D)}, J^{(D)} \right)$ to $(\M_{0,\frac{\g}{\sqrt{2}}}, J)$, which follows from the fact that the Jacobi matrix $d_{(d,c)} \Theta^{(D)}$ of $\Theta^{(D)}$ satisfies the required condition (cf. \cite{arkone})
\begin{equation} \label{symplectjacobitheta}
	\left( d_{(d,c)} \Theta^{(D)} \right)^T \cdot J \cdot d_{(d,c)} \Theta^{(D)} = J^{(D)}
\end{equation}
for any $(d,c) \in \Mdc^{(D)}$, where $J$ is the Poisson structure (\ref{jdef}) of the periodic lattice. That (\ref{symplectjacobitheta}) holds can be checked by a direct computation. It follows that Fix$(S)$ is a symplectic submanifold of $\M_{0,\frac{\g}{\sqrt{2}}}$.

As already mentioned, it follows from Theorem \ref{xkxnkproof} that under the Birkhoff map $\Phi$ of the periodic lattice, Fix$(S)$ is mapped to Fix$(\sz)$, the fixed point set of the map $\sz$ introduced in (\ref{szdef}). We therefore introduce a parametrization of the fixed point set Fix$(\sz)$ of the map $\sz$. Recall that the set $\mathcal{Z} \subset \C^{2N-2}$ is defined by the conditions (\ref{complexdef}). Analogously, we define the set
        \[ \zdir := \left\{ (\z_k)_{1 \leq |k| \leq N'} \in \C^{2N'} \big| \, \overline{\z}_k = \z_{-k} \quad \forall \, 1 \leq k \leq N' \right\},
\]
endowed with the \emph{canonical} symplectic structure 
induced from $\C^{2N'}$. Now consider the embedding
\begin{equation} \label{thetazdef}
\begin{array}{clll}
 \thmz: & \zdir & \to & \mathcal{Z} \\
 & (\z_k)_{1 \leq |k| \leq N'} & \mapsto & \frac{1}{\sqrt{2}} \left( (\z_k)_{1 \leq |k| \leq N'}, (0,0), \left(-e^{4 \pi i k / N} \z_{N'+1-k} \right)_{1 \leq |k| \leq N'} \right)
\end{array}
\end{equation}
Note that
\begin{equation} \label{thmzzdirfixsz}
 \thmz\left( \zdir \right) = Fix(\sz),
\end{equation} 
i.e. $\thmz$ is a parametrization of Fix$(\sz)$.

\begin{proof}[Proof of Theorem \ref{dirtodathm}]
In order to construct a map $\Phi^{(D)}: \M^{(D)} \to \mathcal{P}^{(D)}$ for the Toda lattice with Dirichlet boundary conditions with the claimed properties, we first define a map $\Phi_{\d,\g}^{(D)}: \Mdc^{(D)} \to \zdir$ for any $\d \in \R$ and $\g > 0$. To this end, recall from Theorem \ref{pertodabirkhoffthm} the Birkhoff map $\Phi: \M \to \P$ of the periodic lattice and its restrictions $\Phi_{\b,\a}$ to the level sets $\Mba$ of the Casimirs $C_1$ and $C_2$ for any $\b \in \R$ and $\a > 0$. We define $\Phi_{\d,\g}^{(D)}: \Mdc^{(D)} \to \zdir$ as
\begin{equation} \label{phidgdirdef}
 \Phi_{\d,\g}^{(D)} := \Theta_Z^{-1} \circ \Phi_{0,\frac{\g}{\sqrt{2}}} \circ \Thdc^{(D)},
\end{equation} 
We first check that $\Phi_{\d,\g}^{(D)}$ ist well-defined. Recall from Lemma \ref{flaschkalemma} (i) that 
 \begin{equation} \label{imthdcd}
\textrm{Im}\left( \Thdc^{(D)} \right) = \M_{0,\frac{\g}{\sqrt{2}}} \cap \textrm{Fix}(S).
\end{equation} 
Further, by Theorem \ref{xkxnkproof}, 
\begin{equation} \label{fixszphifixs}
\textrm{Fix}(\sz) = \Phi (\textrm{Fix}(S)).
\end{equation} 
The identities (\ref{imthdcd}) and (\ref{fixszphifixs}), together with (\ref{thmzzdirfixsz}) and the fact that $\thmz$ is one-to-one, imply that $\Phi_{\d,\g}^{(D)}$ is well-defined. Put differently, we can extend (for $\b=0$, $\a = \frac{\g}{\sqrt{2}}$) our previous commutative diagram in the following way:
\begin{displaymath}
 \xymatrixcolsep{3pc}
\xymatrix{\mathcal{Z} & \Mba \ar@/^/[dr]^H \ar[l]_{\Phi_{\b,\a}} &  \\ 
           \zdir \ar[u]^{\Theta_Z} & \Mdc \ar[l]|-{\Phi_{\d,\g}^{(D)}} \ar[r]|-{\HD} \ar[u]|-{\Theta^{(D)}} 
           & \R }
\end{displaymath}
Further note that the map $\Phi_{\d,\g}^{(D)}$ is bijective, which follows from the bijectivity of $\Phi_{\b,\a}$. Moreover, we can now write
\begin{equation} \label{hdcircphi}
 H^{(D)} \circ \left( \Phi_{\d,\g}^{(D)} \right)^{-1} = H \circ \Phi_{\b,\a}^{-1} \circ \Theta_Z.
\end{equation}

We extend $\Phi_{\d,\g}^{(D)}$, given by (\ref{phidgdirdef}), to a map on the entire phase space $\M^{(D)}$ by
\begin{equation} \label{omegadirdeftotal}
\begin{array}{ccll}
 \Phi^{(D)}: & \left( \M^{(D)}, J^{(D)} \right) & \to & \left( \P^{(D)}, J_0^{(D)} \right) \\
 & (d,c) & \mapsto & \left( \Phi_{\d,\g}^{(D)}(d,c), E_1,E_2 \right)
\end{array}
\end{equation}
Note that by (\ref{phidgdirdef}), for any $(x,y,\d,\g) = ((x_n,y_n)_{1 \leq n \leq N'}, \d, \g) \in \P^{(D)}$, the inverse of $\Phi^{(D)}$, as defined by (\ref{omegadirdeftotal}), can be written as
\begin{equation} \label{phidirinv}
 \left( \Phi^{(D)} \right)^{-1}(x,y,\d,\g) = \left( \left( \Thdc^{(D)} \right)^{-1} \circ \Phi_{0,\frac{\g}{\sqrt{2}}}^{-1} \circ \Theta_Z \right)(x,y)
\end{equation}
The bijectivity of $\Phi^{(D)}$ follows from the representation (\ref{phidirinv}) of $\left( \Phi^{(D)} \right)^{-1}$. 
That $\Phi^{(D)}$ is a diffeomorphism follows from the fact that $\Phi_{\d,\g}^{(D)}$ is a diffeomorphism, and that $\Thdc^{(D)}$ and $\thmz$ are diffeomorphisms onto their images. The same holds for the real analyticity property of $\Phi^{(D)}$. In the same way one proves that $\Phi^{(D)}$ is canonical.

To show that the transformed Toda Hamiltonian $\hat{H}^{(D)} = H^{(D)} \circ \left( \Omega^{(D)} \right)^{-1}$ is a function of the actions $(x_n^2 + y_n^2)/2$ and $E_1, E_2$ alone, one combines the analogous property of the Birkhoff map $\Phi$ for the periodic lattice and the fact that for any $(\z_k)_{\1N1} \in \textrm{Fix}(\sz)$, the actions in $\mathcal{Z}$, $I_k = \z_k \z_{-k}$, only depend on the values of the respective actions in $\zdir$, $I_k^{(D)} = \z_k^{(D)} \z_{-k}^{(D)}$, if $\z = \thmz\left( \z^{(D)} \right)$.
\end{proof}

\section{Perturbations of the Dirichlet Toda lattice} \label{dirkamnekh}

We now prove the result on the Birkhoff normal form of the Dirichlet Toda lattice, Theorem \ref{dirtodabnf4}, which is the basis for the nondegeneracy and perturbation results given by Theorems \ref{dirtodahessian} and \ref{kamnekhdirtoda}. 
Instead of imitating the procedure used in \cite{ahtk3} for the periodic lattice, we use the results for the periodic lattice and transform them into results for the Dirichlet lattice by the embedding explained in the previous section, in particular equation (\ref{hdcircphi}).

\begin{proof}[Proof of Theorem \ref{dirtodabnf4}]
On the phase space $\Mba$ of the periodic lattice, it turns out to be helpful to consider relative coordinates $(v,u) = (v_k,u_k)_{1 \leq k \leq N-1}$. They are given by (see \cite{ahtk3} for details, in particular for a proof  that this is a canonical transformation)
\begin{equation} \label{relcoorddef}
 v_n := q_{n+1} - q_n, \qquad u_n := n\b - \sum_{j=1}^n p_k \quad (1 \leq n \leq N-1)
\end{equation}
In terms of these coordinates, $H$ (i.e. the Hamiltonian of the periodic lattice) is given by
\begin{displaymath} 
  \Hba \!=\! \frac{N\b^2}{2} \!+\! \frac{1}{2} \left( \! u_1^2 \!+\! \sum_{l=1}^{N-2} (u_l \!-\! u_{l\!+\!1})^2 \!+\! u_{N-1}^2 \! \right) \!+\! \a^2 \left( \! \sum_{k=1}^{N-1} e^{-v_k} + e^{\sum_{k=1}^{N-1} v_k} \! \right).
\end{displaymath}
Note that $\Hba$ only depends on the $2N-2$ variables $(v_k,u_k)_{1 \leq k \leq N-1}$, unlike $H$, which depends on the $2N$ variables $(b_j,a_j)_{1 \leq j \leq N}$.

The map $\Oba$ defined by
\begin{equation} \label{obadef}
\begin{array}{cccccc}
  \Oba: & \R^{2N-2} (\cong \mathcal{Z}) & \to & \Mba & \to & \R^{2N-2} \\
& (x_k, y_k)_{\1N1} & \mapsto & (b_j, a_j)_{1 \leq j \leq N} & \mapsto & (v_k, u_k)_{\1N1}, \end{array}
\end{equation}
i.e. the composition of the Birkhoff map of the periodic lattice with the transformation defined in (\ref{relcoorddef}) is then a canonical real analytic transformation, as both $(x,y)$ and $(v,u)$ are canonical coordiantes for $\Mba$.



With the map $\Oba$, (\ref{hdcircphi}) now reads
\begin{equation} \label{hdcircphiomegabahba}
 H^{(D)} \circ \left( \Phi_{\d,\g}^{(D)} \right)^{-1} = \Hba \circ \Oba \circ \Theta_Z,
\end{equation} 
 and we obtain the following diagram:
\begin{displaymath}
 \xymatrixcolsep{3pc}
\xymatrix{ & \R^{2N-2} \ar@/^1pc/[rdd]^{\Hba} & \\
\mathcal{Z} \ar@/^1pc/[ur]^{\Oba} & \Mba \ar@/^/[dr]|-H \ar[l]|-{\Phi_{\b,\a}} \ar[u]|-{(v,u)} &  \\ 
           \zdir \ar[u]^{\Theta_Z} 
           & \Mdc \ar[l]|-{\Phi_{\d,\g}^{(D)}} \ar[r]|-{\HD} \ar[u]|-{\Theta^{(D)}} 
           & \R }
\end{displaymath}



Next, we recall a map which in the periodic case puts the Hamiltonian $\Hba$ into Birkhoff normal form up to order $2$, namely the Jacobian of the map $\Oba$,
\begin{equation} \label{rbadef}
  \Rba: \R^{2N-2} \to \R^{2N-2}, \quad (\xi, \eta) \mapsto (v,u) = d_{x,y} \Oba|_{(x,y)=(0,0)} (\xi, \eta)
\end{equation}

Precisely, we first rewrite the identity (\ref{hdcircphiomegabahba}) in a manner suitable to our purposes:
\begin{eqnarray}
H^{(D)} \circ \left( \Phi_{\d,\g}^{(D)} \right)^{-1} & = & \Hba \circ \Oba \circ \Theta_Z \nonumber\\
 & = & \Hba \circ \Rba \circ (\Rba^{-1} \circ \Oba) \circ \Theta_Z \nonumber\\
& = & \Hba \circ \Rba \circ \Theta_Z \circ \underbrace{\Theta_Z^{-1} \circ (\Rba^{-1} \circ \Oba) \circ \Theta_Z}_{=: F_1} \label{hdphidf1def}
\end{eqnarray}
To make sure that the definition of the map $F_1$ makes sense, we have to check that $\Rba^{-1} \circ \Oba$ maps Im$(\Theta_Z)$ onto itself, so that $\Theta_Z^{-1}$ is defined on Im$(\Rba^{-1} \circ \Oba \circ \Theta_Z)$. This can be checked by using explicit formulas for $\Rba$ (see \cite{ahtk3}). Further, since $\Rba^{-1} \circ \Oba$ is of the form ``Id $+$ higher order terms'', the same holds for $F_1$ (recall from (\ref{thetazdef}) that the pullback of any $\zeta_k$ with $1 \leq k \leq \D$ under $\Theta_Z$ is $\zeta_k / \sqrt{2}$).

Graphically, this leads to the following situation:
\begin{displaymath}
 \xymatrixcolsep{3pc}
\xymatrix{ && \R^{2N-2} \ar@/^1pc/[rdd]^{\Hba} & \\
\mathcal{Z} \ar@/^1pc/[urr]^{\Rba} & \mathcal{Z} \ar@/^/[ur]|-{\Oba} & \Mba \ar@/^/[dr]|-H \ar[l]|-{\Phi_{\b,\a}} \ar[u]|-{(v,u)} &  \\ 
           \zdir \ar[u]^{\Theta_Z} & \zdir \ar[l]^{F_1} \ar[u]|-{\Theta_Z} 
           & \Mdc \ar[l]|-{\Phi_{\d,\g}^{(D)}} \ar[r]|-{\HD} \ar[u]|-{\Theta^{(D)}} 
           & \R }
\end{displaymath}
It follows from (\ref{hdphidf1def}) that $H^{(D)} \circ \left( \Phi_{\d,\g}^{(D)} \right)^{-1} \circ F_1^{-1}$ is in Birkhoff normal form up to order $2$.

Further, we use the map $\Xi$ which transforms the periodic lattice into Birkhoff normal form up to order $4$, which we pull back to the (transformed) phase space $\mathcal{Z}$ of the Dirichlet lattice by $\Theta_Z$,
\begin{displaymath}
 \Hba \circ \Rba \circ \Xi \circ \Theta_Z = \Hba \circ \Rba \circ \Theta_Z \circ \underbrace{\Theta_Z^{-1} \circ \Xi \circ \Theta_Z}_{=:F_2}.
\end{displaymath}

To prove that the map $F_2$ is well-defined, one again has to show that $\Xi$ maps Im$(\Theta_Z)$ onto itself; for this proof, we refer to \cite{ahtk5}, where we proved this in the context of more general FPU chains. Furthermore, $F_2$ is of the form ``Id $+$ higher order terms'', which follows from the corresponding property of $\Xi$ and the definition of $\Theta_Z$. Graphically, the situation is now as follows:
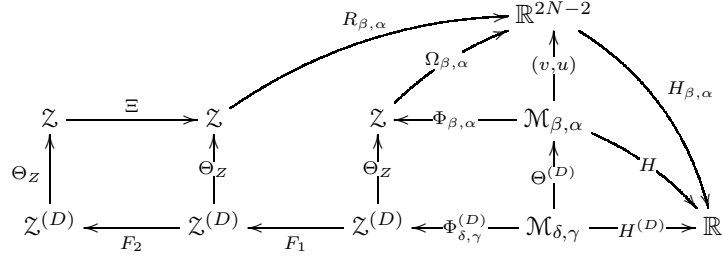
\begin{displaymath}
 \xymatrixcolsep{3pc}
\xymatrix{ &&& \R^{2N-2} \ar@/^1pc/[rdd]^{\Hba} & \\
\mathcal{Z} \ar[r]^{\Xi} & \mathcal{Z} \ar@/^1pc/[urr]^{\Rba} & \mathcal{Z} \ar@/^/[ur]|-{\Oba} & \Mba \ar@/^/[dr]|-H \ar[l]|-{\Phi_{\b,\a}} \ar[u]|-{(v,u)} &  \\ 
          \zdir  \ar[u]^{\Theta_Z} & \zdir \ar[u]|-{\Theta_Z} \ar[l]^{F_2} & \zdir \ar[l]^{F_1} \ar[u]|-{\Theta_Z} 
          & \Mdc \ar[l]|-{\Phi_{\d,\g}^{(D)}} \ar[r]|-{\HD} \ar[u]|-{\Theta^{(D)}} 
          & \R }
\end{displaymath}

Recall from Theorem \ref{dirtodathm} that $H^{(D)} \circ \left( \Phi_{\d,\g}^{(D)} \right)^{-1}$ is in Birkhoff normal form. By the commutativity of the above diagram,
\begin{displaymath}
 H^{(D)} \circ \left( \Phi_{\d,\g}^{(D)} \right)^{-1} = \Hba \circ \Rba \circ \Theta_Z \circ F_1.
\end{displaymath}
On the other hand, $\Hba \circ \Rba \circ \Xi \circ \Theta_Z$ is locally around the origin in Birkhoff normal form up to order $4$, and again by commutativity,
\begin{displaymath}
 \Hba \circ \Rba \circ \Xi \circ \Theta_Z = \Hba \circ \Rba \circ \Theta_Z \circ F_2^{-1}.
\end{displaymath}
Since $F_1$ and $(F_2)^{-1}$ are both of the form ``Id $+$ higher order terms'', by the uniqueness of the Birkhoff normal form (see e.g. \cite{kapo}), the expansion of $\Hba \circ \Rba \circ \Theta_Z$ in the two coordinate systems $F_1$ and $(F_2)^{-1}$ coincide up to order $4$. Therefore, $\Hba \circ \Rba \circ \Xi \circ \Theta_Z$ provides us directly with the desired normal form by taking the pullback of the normal form $\Hba \circ \Rba \circ \Xi$ of the periodic lattice with respect to the embedding $\Theta_Z$. The expansion of $\Hba \circ \Rba \circ \Xi$ with respect to the $I_k$'s up to order $2$ (i.e. up to order $4$ in the $(x_k,y_k)$'s) is in the case $\b=0$ given by (see \cite{ahtk3})
\begin{displaymath}
 N \a^2 + 2\a \sum_{k=1}^{N-1} s_k I_k + \frac{1}{4N} \sum_{k=1}^{N-1} I_k^2 + O(I^3).
\end{displaymath}
Using $I_k = I_{N-k}$ and $I_{N/2} = 0$ on Fix$(S)$ (see Corollary \ref{fixsspectrumlemma}) as well as $s_k = s_{N-k}$ and replacing $\a$ by $\g/\sqrt{2}$ and the $I_k$'s by their pullback $I_k/2$ with respect to $\Theta_Z$, we obtain
\begin{displaymath} 
 \frac{N \g^2}{2} + \sqrt{2} \, \g \sum_{k=1}^{\D} s_k I_k + \frac{1}{16(\D+1)}\sum_{k=1}^{\D} I_k^2 + O(I^3),
\end{displaymath}
which is the claimed expansion (\ref{dirtodabnf4formula}). This proves Theorem \ref{dirtodabnf4}.
\end{proof}

\section*{Acknowledgements}

It is a pleasure to thank Percy Deift and Thomas Kappeler for valuable discussions. The support of SNF and of a very generous sponsor is also gratefully acknowledged.


\appendix



\textsc{Institut f\"ur Angewandte Mathematik und Physik, ZHAW School of Engineering, Technikumstrasse 9, Postfach, CH-8401 Winterthur, Switzerland}

\emph{E-mail address:} \tt{andreas.henrici@zhaw.ch}

\end{document}